\documentclass[10pt,twocolumn,twoside]{IEEEtran} 

\IEEEoverridecommandlockouts                              

\usepackage{graphics} 
\usepackage{epsfig} 
\usepackage{amsmath} 

\usepackage{amssymb,amsthm}  
\usepackage{amscd}
\usepackage[noadjust]{cite}
\usepackage{float}
\usepackage{times}

\newtheorem{Definition}{Definition}
\newtheorem{theorem}{Theorem}
\newtheorem{Remark}{Remark}
\newtheorem{pro}{Proposition}
\newtheorem{lem}{Lemma}

\newcommand{\R}{\mathbb{R}}
\newcommand{\ol}{\overline}
\renewcommand{\cal}{\mathcal}
\renewcommand{\(}{\left (}
\renewcommand{\)}{\right )}

\renewcommand{\;}{\,;\,}
\newcommand{\N}{\mathbb{N}}

\title{\LARGE \bf
Swarm Aggregation under Fading Attractions}

\author{Xudong Chen$^*$
\thanks{$^*$X. Chen is with the Coordinated Science Laboratory, University of Illinois at Urbana-Champaign. email: xdchen@illinois.edu.}%
}

\begin{document}

\maketitle
\thispagestyle{empty}
\pagestyle{empty}

\begin{abstract}
Gradient descent methods have been widely used for organizing multi-agent systems, in which they can provide decentralized control laws with provable convergence. 
Often, the control laws are designed so that two neighboring agents repel/attract each other at a short/long distance of separation. 
When the interactions between neighboring agents are moreover nonfading, the potential function from which they are derived is radially unbounded. Hence, the LaSalle's principle is sufficient to establish the system convergence. 
This paper investigates, in contrast, a more realistic scenario where interactions between neighboring agents have fading attractions. In such setting, the LaSalle type arguments may not be sufficient. To tackle the problem,  
we introduce a class of partitions, termed \emph {dilute partitions}, of formations which cluster agents according to the inter- and intra-cluster interaction strengths.  We then apply dilute partitions to trajectories of formations  generated by the multi-agent system, and show that each of the trajectories remains bounded along the evolution,  and converges to the set of equilibria.    
\end{abstract}

\section{Introduction}
The use of gradient descent for organizing a group of mobile autonomous agents has been widely appreciated in mathematics and in its real-world applications. Descent equations often provide the most direct demonstration of the existence of local minima, and provide easily implemented algorithm for finding the minima. Furthermore,  in the context of multi-agent control, gradient descent can be interpreted as providing decentralized  control laws for pairs of neighboring agents in the system. Specifically, we consider a class of multi-agent systems in which pairs of neighboring agents  attract/repel each other in a reciprocal way, depending {\it only} on the distances of separation. Then, the resulting dynamics of the agents evolve as a gradient flow over a Euclidean space. We describe below the model in precise terms:   
\vspace{3pt}
\\
{\bf Model}. 
Let $G = (V,E)$ be an undirected connected graph of~$N$ vertices, with $V = \{v_1,\ldots, v_N\}$ the vertex set, and $E$ the edge set. We denote by $(v_i,v_j)$ an edge of $G$. 
Let $V_i$ be the set of neighbors of~$v_i$, i.e., $$V_i = \{v_j\in V\mid (v_i,v_j) \in E\}.$$   To each vertex~$v_i$, we assign an agent~$i$, with $ x_i \in \R^n$ its coordinate. With a slight abuse of notation, we refer to agent~$i$ as~$x_i$. For every edge $(v_i,v_j) \in E$, we let $d_{ij}$ be the distance between $x_i$ and $x_j$, i.e., $d_{ij}:= \| x_i- x_j\|$.  
The equations of motion of the $N$ agents $ x_1,\cdots, x_N$ in $\mathbb{R}^n$ are given by
\begin{equation}\label{MODEL}
\frac{d}{dt}{ x}_i=\sum_{v_j\in V_i}g_{ij}(d_{ij})(x_j-x_i),\hspace{10pt} \forall\, v_i \in V.
\end{equation}
Each scalar function $g_{ij}$ is assumed to be continuously differentiable; we refer to~$g_{ij}$, for $(v_i,v_j)\in E$, the {\bf interaction functions} associated with system~\eqref{MODEL}.  
An important property associated with system~\eqref{MODEL} is that the dynamics of the agents evolve as a gradient flow. A direct computation yields that the associated potential function is given by 
\begin{equation}\label{POTENTIAL}
\Psi(x_1,\ldots,x_N):=\sum_{(v_i,v_j)\in E}\int_1^{\|x_j - x_i\|} s g_{ij}(s)ds.
\end{equation}

Designing of the interaction functions that are necessary for organizing such multi-agent system has been widely investigated: questions about swarm aggregation and avoidance of collisions~\cite{GP,chu2003self,XC2014ACC}, questions about local/global stabilization of targeted configurations~\cite{chu2003self,krick2009,dimarogonas2008stability,xudongchen2015CDCtriangulatedformationcontrol,zhiyongsun2015ECC}, questions about robustness issues of control laws under perturbations~\cite{AB2012CDC,sun2014CDC,USZB,mou2014CDC}, questions about counting number of critical formations~\cite{BDO2014CT,UH2013E}  have all been treated to some degree. 
We also refer to~\cite{xudongchen2015CDCformationcontroltimevaryinggraph,XC2014CDC,JB2006cdc,cao2008control,baillieul2007combinatorial,AB2013TAC,AL2014ECC,lin2014distributed,chen2015decentralized} for other types of models for multi-agent control, as variants of system~\eqref{MODEL}. 
  
For the purpose of achieving swarm aggregation, the interaction functions $g_{ij}$'s are often designed so that neighboring agents 
attract each other at a long distance. In particular, we note here that if the underlying graph $G$ is connected, and the interaction functions between neighboring agents have  non-fading attractions (as considered in most of the literatures: see, for example,~\cite{GP,chu2003self,zhiyongsun2015ECC,krick2009,dimarogonas2008stability}); then, for any initial condition, the resulting gradient flow will converge to the set of equilibria. In other words, there is no escape of agents to infinity along the evolution of the multi-agent system. Indeed, in any of such case, the  associated potential function~\eqref{POTENTIAL} is {\it radially unbounded}, i.e., it approaches to infinity as the size of a formation tends to infinity. So then, each trajectory  of system~\eqref{MODEL} has to remain bounded, and hence converges to the set of equilibria.  

On the other hand, it is more realistic  to assume that the magnitude of an attraction between two neighboring agents fades away as their mutual distance grows. 
We refer to~\cite{cucker2007emergent}, as an example, for modeling the flocking behavior with fading interactions. Specifically, the authors there considered a second order model: 
$$
\left\{
\begin{array}{l}
\dot x_i = v_i\\
\dot v_i = \sum^N_{j = 1} g(d_{ij}) (v_j - v_i),  
\end{array}
\right.
$$
with the graph $G$ being complete and without repulsions, i.e., the function $g(d)$ is positive at all distances $d>0$.  
Also, we recall that the Lennard-Jones force, which describes the interaction between a pair of neutral molecules/atoms, has strong repulsion and fading attraction. 
 
 We note here that,  under the assumption of fading attraction, the potential function associated with system~\eqref{MODEL} may remain bounded as the size of a formation grows; indeed, one may find a continuous path of formations along which the potential function decreases while the size of formation approaches to infinity. In particular,  conventional techniques for proving convergence of gradient flows, such as using the potential function as a Lyapunov function and then appealing to the LaSalle's principle~\cite{lasalle1960some}, may not work in this case. 
Nevertheless, we are still able to show that all the trajectories generated by system~\eqref{MODEL} converge to the set of equilibria. The proof of the system convergence  relies on the use of a class of partitions, termed {\it dilute partitions}, of formations introduced in section~III. Roughly speaking, dilute partitions decompose formations into different clusters of agents according to certain combinatorial and metric conditions. 
We apply dilute partitions to trajectories of formations generated by system~\eqref{MODEL}, and investigate how clusters of agents evolve over time and interact with each other. In particular, we show that each trajectory generated by system~\eqref{MODEL} has to remain bounded, and hence converges to the set of equilibria.  
This approach, via the use of dilute partition, to multi-agent systems might be of independent interest for studying other problems that involve large sized formations.


This paper expands on some preliminary result presented in~\cite{XC2014ACC} by, among others, providing an analysis of system~\eqref{MODEL} with an arbitrary connected graph (whereas in~\cite{XC2014ACC}, we dealt only with the complete graph), a finer description of the dilute partitions and the associated properties, and a considerable amount of analyses and proofs that were left out. 
The remainder of the paper is organized as follows. In section~II, we introduce definitions and notations and describe some preliminary results about system~\eqref{MODEL}. We also state the main theorem of the paper. In particular, the main theorem states that  the equilibria of system~\eqref{MODEL} have bounded size, and moreover, all trajectories generated by system~\eqref{MODEL} converge to the set of equilibria under the assumption of fading attractions. Sections~III and~IV are devoted to establishing properties of system~\eqref{MODEL} that are needed for proving the main theorem. A detailed organization of these two sections will be given after the statement of the theorem. We provide conclusions in the last section. The paper ends with Appendices containing proofs of some technical results.

\section{Backgrounds and Main Theorem}
In this section, we introduce the main definitions used in this work, describe some preliminary results, and state the main theorem of the paper. 

\subsection{Backgrounds and notations}
Let $G = (V,E)$ be an undirected graph of $N$ vertices. Let $V'$ be a subset of $V$; a subgraph $G' = (V', E')$ of $G$ is said to be {\it induced} by $V'$ if the following condition is satisfied: an edge $(v_i,v_j)$ is in $E'$ if and only if $(v_i,v_j)$ is in $E$.  

Given a formation of $N$ agents in $\R^n$, with states $x_1,\ldots, x_N$,  respectively, we set $p := (x_1,\ldots,x_N) \in \R^{nN}$. We call $p$ a {\bf configuration};  
a configuration $p \in P_G$ can be viewed as an {\it embedding} of the graph $G$ in $\R^n$ by assigning vertex $v_i$ to $x_i$. We call the pair $(G, p)$ a {\bf framework}. 
We define the  {\it configuration space}  $P_{G}$, associated with the graph~$G$, as follows:
\begin{equation*}
P_{G}:=\left\{( x_1,\cdots, x_N)\in \mathbb{R}^{nN} \mid  x_i\neq  x_j, \, \forall \, (v_i,v_j)\in E \right\}. 
\end{equation*} 
Equivalently, $P_{G}$ is the set of embeddings of the graph $G$ in $\mathbb{R}^n$ whose neighboring vertices have distinct positions. 
Let $(G, p)$ be a framework, with $p = (x_1,\ldots, x_N) \in P_G$.  
 Let $G' = \(V',E' \)$ be a subgraph of $G$, with $V' = \{v_{i_1},\ldots,v_{i_k}\}$. We call $p'\in P_{G'}$ a {\bf sub-configuration} of $p$ associated with $G'$ if $p' = (x_{i_1},\ldots, x_{i_k})$, and correspondingly $(G',p')$ a sub-framework of $(G,p)$.   
\vspace{3pt}
\\ 
{\bf Attraction/Repulsion functions}. We  now introduce the class of interaction functions, termed {\it attraction/repulsion functions},  that are considered in the paper. Roughly speaking, an attraction/repulsion function between a pair of agents is such that the two agents attract/repel each other at a long/short distance. Furthermore, we require that the repulsion go to infinity as the mutual distance between the agents approaches to zero and that the attraction fade away as the distance grows.  
A typical example of such function is the Lennard-Jones type interaction:
\begin{equation}\label{eq:typicalexample}
g(d) = -\frac{\sigma_1}{d^{n_1}} + \frac{\sigma_2}{d^{n_2}} 
\end{equation}
with $\sigma_1$, $\sigma_2$ positive real numbers, and $n_1$, $n_2$ positive integers satisfying $n_1 > n_2 > 1$. 
We now define attraction/repulsion functions in precise terms. Let $\R_+$ be the set of strictly positive real numbers. We denote by $\operatorname{C}(\R_+,\R)$ the set of continuous functions from $\R_+$ to $\R$. We have the following definition:

\begin{Definition}[Attraction/Repulsion functions]\label{def:fadingattraction}
A function $g$ in $\operatorname{C}(\R_+,\R)$ is an {\bf attraction/repulsion function} if $g$ satisfies the following conditions: 
\begin{enumerate}
\item {\bf Strong repulsion}: $$\lim_{d\to 0+} dg(d)=-\infty,$$ 
and moreover,  $$\displaystyle \lim_{d\to 0+}\int^1_d sg(s)ds=-\infty.$$ 
\item {\bf Fading attraction}: There exists a number $\alpha_+ > 0$ such that 
\begin{equation*}
g(d) > 0, \hspace{10pt} \forall \, d \ge \alpha_+,    
\end{equation*} 
and moreover,  $$\lim_{d\to\infty} d g(d)=0.$$   
\end{enumerate}\,
\end{Definition}

\noindent Note that the function $dg_{ij}(d)$ shows up in Definition~\ref{def:fadingattraction} because $| dg_{ij}(d) |$ represents the actual magnitude of attraction/repulsion between $x_i$ and $x_j$. 

 We assume in the remainder of the paper that all the interaction functions $g_{ij}$, for all $(v_i,v_j) \in E$, are attraction/repulsion functions. Furthermore, we assume, without loss of generality, that the positive number $\alpha_+$ in Definition~\ref{def:fadingattraction} can be applied to all $g_{ij}$, i.e., 
 \begin{equation}\label{eq:defx+}
g_{ij}(d) > 0, \hspace{10pt} \forall \, d \ge \alpha_+,   
 \end{equation}
 for all $(v_i,v_j) \in E$. 
 
In the paper, we often deal with sub-systems of~\eqref{MODEL}, especially, the subsystems induced by subgraphs of $G$. We thus have the following definition: 
 

\begin{Definition}[Induced sub-systems]
Let $G' = (V',E')$ be a subgraph of $G$. A multi-agent system is a  {\bf sub-system induced by $G'$}  if it is comprised of  agents $ x_i$, for $v_i\in V'$,  together with the interaction functions   $g_{ij}$,  for $(v_i,v_j)\in E'$. Specifically,  the dynamics of the agents in the induced sub-system are given by:
\begin{equation*}\label{eq:inducedsub-system}
\dot{ x}_{i} = \sum_{v_j\in V'_i}g_{ij}(d_{ij}) ( x_j- x_i),\hspace{10pt} \forall v_i\in V'
\end{equation*}
with $V'_i$ the neighbors of $i$ in $G'$. 
\end{Definition}

For each configuration $p\in P_G$, we denote by $f(p)$ the vector field of system~\eqref{MODEL}. The configuration $p$ is said to be an {\bf equilibrium} of system~\eqref{MODEL} if $f(p) = 0$.  For each $v_i\in V$, we let $f_i(p)\in \R^n$ be defined by restricting $f(p)$ to agent $x_i$, i.e., 
$$
f_i(p):= \sum_{v_j \in V_i} g_{ij}(d_{ij}) ( x_j- x_i).
$$
Similarly, for a subgraph $G' = (V', E')$ of $G$, we denote by $f_{V'}(p)\in \R^{n|V'|}$ the restricting of $f(p)$ to the sub-configuration $p'$ associated with $G'$.  
\subsection{Preliminaries and the main result}

In this subsection, we describe some preliminary results, and then state the main theorem of the paper.  
Recall that the dynamics of system~\eqref{MODEL} is a gradient flow of $\Psi$ defined in~\eqref{POTENTIAL}.    
By assuming that all $g_{ij}$, for $(v_i,v_j) \in E$, are attraction/repulsion functions,  we have the following fact: 

\begin{lem}\label{lem:phiboundedbelow}
The potential function $\Psi: P_G \to \mathbb{R}$ is bounded below, i.e., 
$$
\inf\{\Psi(p) \mid p\in P_G\} > -\infty.
$$ \, 
\end{lem}

\begin{proof}
First, note that from the condition of {\it strong repulsion}, there is a positive number $\alpha_-> 0$ such that 
$$
g_{ij}(d) < 0, \hspace{10pt} \forall \, d \le \alpha_- \mbox{ and } \forall\, (v_i,v_j)\in E.  
$$
We also recall that $\alpha_+$ is defined in~\eqref{eq:defx+} such that 
$$
g_{ij}(d) > 0, \hspace{10pt} \forall \, d \ge \alpha_+ \mbox{ and } \forall\, (v_i,v_j)\in E.  
$$
This, in particular, implies that for all $(v_i,v_j) \in E$, 
$$
\min_{d\in [\alpha_-, \alpha_+]} \int^d_{1} sg_{ij}(s) ds = \inf_{d \in \R_+} \int^d_{1} sg_{ij}(s)ds. 
$$
Now, let
\begin{equation}\label{eq:defpsi0}
\psi_0:= \min_{(v_i,v_j) \in E}\,\left\{ \min_{d\in [\alpha_-, \alpha_+]} \int^d_{1} s g_{ij}(s) ds \right\}; 
\end{equation}
then, we have
$$
\Psi(p) \ge |E|\, \psi_0, \hspace{10pt} \forall\, p\in P_G,
$$
which completes the proof.
\end{proof}

It is well known that along a trajectory of a gradient flow, the potential function is non-increasing. On the other hand, the condition of {\it strong repulsion} implies that the potential function $\Psi$ is infinite if the distance of separation of two neighboring agents is zero. This, in particular, implies that there is no collision of neighboring agents along the evolution, and hence solutions of system~\eqref{MODEL} exist for all time. 
Furthermore, for a configuration $p = (x_1,\ldots, x_N)\in P_G$, if we let $d_-(p)$ and $d_+(p)$ be defined as follows: 
\begin{equation}\label{eq:defd-d+}
\left \{
\begin{array}{l}
d_-(p) := \min \left \{  \|x_j - x_i\| \mid  (v_i, v_j)\in E   \right \} \vspace{3pt} \\
d_+(p) := \max \left \{  \|x_j - x_i\| \mid  (v_i, v_j)\in E   \right \}, 
\end{array}
\right.
\end{equation} 
then we have the following fact:

\begin{lem}\label{ELB} 
Let $p(0)\in P_G$ be the initial condition of system~\eqref{MODEL}, and $p(t)$ be the trajectory generated by the system. Then, 
$$
\inf \left\{  d_-(p(t))  \mid t\ge 0 \right \} > 0. 
$$\,
\end{lem}

\begin{proof}
Let $\psi_0$ be defined in~\eqref{eq:defpsi0}. Then, from the condition of {\it strong repulsion}, there exists a number $d > 0$ such that
$$
\int^d_{1}  s g_{ij}(s) ds + \(|E| - 1\)\, \psi_0 > \Psi(p(0))
$$
for all $(v_i,v_j) \in E$. We now show that 
$
d_-(p(t))  > d$ for all  $t \ge 0$.  
Suppose that, to the contrary, there exists an instant $t\ge 0$ such that $\|x_j(t) - x_i(t)\| = d$ for some $(v_i,v_j)\in E$. Then, by definition of $\psi_0$, we have
$$
\Psi(p(t)) \ge  \int^d_{1}  sg_{ij}(s) ds + \(|E| - 1\)\, \psi_0 > \Psi(p(0))
$$
which contradicts the fact that $\Psi(p(t))$ is non-increasing in~$t$. This completes the proof.
\end{proof}

Note that from Lemmas~\ref{lem:phiboundedbelow} and~\ref{ELB}, if the potential function~$\Psi$ is such that
\begin{equation}\label{eq:potentialconferencepaper}
\lim_{d_+(p) \to \infty} \Psi(p) = \infty, 
\end{equation}
then each trajectory $p(t)$ of system~\eqref{MODEL} has to remain bounded, and hence converges to the {set of  equilibria}.  Yet,~\eqref{eq:potentialconferencepaper} may not hold under the condition of {\it fading attraction}. For example, if each $g_{ij}$ is a Lennard-Jones type interaction, i.e.,
$$
g_{ij}(d) = - \frac{\sigma_{ij,1}}{d^{n_{ij,1}}} + \frac{\sigma_{ij,2}}{d^{n_{ij,2}}} 
$$
with $n_{ij,1} > n_{ij,2} > 2$; then,  for all $(v_i,v_j) \in E$, we have
$$
\int^{\infty}_{1} sg_{ij}(s)ds = -\frac{\sigma_{ij,1}}{n_{ij,1} - 2} + \frac{\sigma_{ij,2}}{n_{ij,2} - 2} < \infty, 
$$
and hence $\Psi(p)$ may remain bounded as $d_+(p)$ diverges. Nevertheless, we are still able to establish the convergence of system~\eqref{MODEL}. 
We state below  the main result of the paper. 

\begin{theorem}\label{thm:MAIN}
Let $G = (V, E)$ be a connected undirected graph, and let $g_{ij}$, for $(v_i,v_j) \in E$,  be attraction/repulsion functions. Then, the multi-agent system~\eqref{MODEL} satisfies the following properties:
\begin{enumerate}
\item There exist two positive numbers $D_-$ and $D_+$ such that if $p$ is an equilibrium of system~\eqref{MODEL}, then 
$$
D_- \le d_-(p) \le d_+(p) \le D_+
$$
with $d_-(p)$ and $d_+(p)$ defined in~\eqref{eq:defd-d+}.
\item For any initial condition $p(0)\in P_G$, the trajectory $p(t)$ of system~\eqref{MODEL} converges to the set of equilibria. 
\end{enumerate}\, 
\end{theorem}


In the remainder of the paper, we establish properties of system~\eqref{MODEL} that are needed to prove Theorem~\ref{thm:MAIN}. 
In section~\ref{sec:dp}, we introduce a class of partitions, termed {\it dilute partitions}, of frameworks, which decomposes frameworks into disjoint sub-frameworks satisfying certain 
combinatorial and metric properties. 
This is a rich question, related to the $k$-means clustering~\cite{macqueen1967some} and its variants.  We then apply dilute partitions to unbounded sequences of frameworks,  and describe relevant properties associated with it. 
 In section~\ref{sec:cgf}, we apply dilute partitions to frameworks  along a class of trajectories generated by system~\eqref{MODEL}, and establish certain path behavior of the trajectories, which is relevant to the proof of system convergence.  




\section{Dilute Partitions and Diluting Sequences}\label{sec:dp}

Let $(G, p)$ be a framework. We say that $\sigma = \{(G_i, p_i)\}^m_{i=1}$, with $G_i = (V_i, E_i)$,  is a {\bf partition} of $(G,p)$ if $\sigma$ satisfies the following conditions: 
\begin{enumerate}
\item The subsets $\{V_1,\ldots,v_m\}$ of $V$ form a partition:  
\begin{equation}\label{eq:inducedpartitionofV}
V = \sqcup^m_{i=1} V_i.
\end{equation}
\item Each  $G_i$ is a subgraph of $G$ induced by $V_i$, and each $p_i$ is a sub-configuration of $p$ associated with $G_i$. 
\end{enumerate}
We refer to~\eqref{eq:inducedpartitionofV} the partition of $V$ {\it induced by $\sigma$}. 

Now, for a partition $\sigma = \{(G_i, p_i)\}^m_{i=1}$ of a framework $(G,p)$,  let the {diameter} of a sub-configuration $p_i$ be 
$$\phi(p_i) : = \max \left \{\| x_k - x_j\| \mid v_j, v_k \in V_i \right\}. $$ 
 We then define the {\bf intra-cluster distance} of $\sigma$ by
\begin{equation}\label{eq:intradistance}
\cal{L}_-(\sigma): = \max\left \{ \phi(p_i) \mid 1\le i \le m  \right\}.
\end{equation}
Given two distinct sub-frameworks $(G_i,p_i)$ and $(G_j, p_j)$, let $d(p_i,p_j)$ be the distance between $p_i$ and $p_j$: 
$$
d(p_i,p_j): = \min\left\{\|x_{i'} - x_{j'}\| \mid  v_{i'} \in V_i, \,  v_{j'} \in V_j \right \}. 
 $$
We say that $(G_i,p_i)$ and $(G_j, p_j)$ are {\bf adjacent} if there is an edge $(v_{i'},v_{j'})$ of $G$, with $v_{i'}\in V_i$ and $v_{j'} \in V_j$.  We then define the {\bf inter-cluster distance} of $\sigma$ by 
\begin{equation}\label{eq:interdistance}
\cal{L}_+(\sigma) := \min_{(i,j)} \left \{ d(p_i,p_j) \right \}, 
\end{equation} 
where the minimum is taken over the pairs $(i,j)$ for $(G_i,p_i)$ and $(G_j,p_j)$ to be adjacent. 
With the definitions and notations above, we define dilute partitions:

\begin{Definition}[Dilute partitions]\label{def:dilutepartition}
Let $(G, p)$ be a framework, with $G$ connected and $p\in P_G$. 
A partition $\sigma = \{(G_i, p_i)\}^m_{i=1}$ of $(G,p)$ is a \textbf{dilute partition} with respect to a positive number~$l$ if it satisfies the following two conditions:
\begin{enumerate}
\item Each $G_i$ is connected. 
\item If $(G_i,p_i)$ and $(G_j,p_j)$ are adjacent, then  
$d(p_i,p_j) > l$ and ${\max\{\phi(p_i), \phi(p_j)\}} < {d(p_i,p_j)}$.  
\end{enumerate}\,
\end{Definition}



In the remainder of the section, we fix a connected graph $G$, and assume that $G$ has at least two vertices. 
For a positive number $l$ and a configuration $p\in P_G$, we let $\Sigma(l\; p)$ be the set of dilute partitions of $(G,p)$ with respect to $l$.  
Note that $\Sigma(l\; p)$ is nonempty because $\Sigma(l\; p)$ always contains the {\bf trivial partition}, namely,  the partition which has only one cluster containing all the agents. A partition of $(G,p)$ is said to be {\it nontrivial} if it is not the trivial partition. 
We also note that from Definition~\ref{def:dilutepartition}, if $\sigma \in \Sigma(l,p)$ and $l \ge l' >0$, then $\sigma \in \Sigma(l',p)$. In other words, we have $\Sigma(l,p) \subseteq \Sigma(l',p)$ for all configurations $p\in P_G$.

We will now state the main result of the section, which relates dilute partitions to sequences of diverging configurations: 
\vspace{3pt}
\noindent
{\bf Diluting sequence}. 
Let $\{p(i)\}_{i\in\N}$ be a sequence of configurations in $P_G$. We say that $\{p(i)\}_{i\in\N}$ is {\bf unbounded} if for any $d > 0$, there exists an $i\in\N$ such that $\phi(p(i)) > d$.  
We now formalize in detail the following fact: for any unbounded sequence $\{p(i)\}_{i\in\N}$, there is a subsequence $\{p(n_i)\}_{i\in\N}$ such that (i) the agents in $p(n_i)$, for $i\in\N$, are clustered in the same way; (ii) the inter-cluster distances diverge while the intra-cluster distances remain bounded. Precisely, we state the following result:


\begin{theorem}[Diluting sequence]\label{CCOSODC} 
Let $\{p(i)\}_{i\in\N}$ be an unbounded sequence in $P_G$,  
and $\{l_i\}_{i\in \N} $ be a sequence of positive real numbers, with $\lim_{i\to \infty} l_i = \infty $. 
Then, there is a subsequence $\{p(n_i)\}_{i\in\mathbb{N}}$ out of $\{p(i)\}_{i\in\N}$, together with a sequence of nontrivial dilute partitions 
$$\left \{\sigma_i\in \Sigma(l_i\; p(n_i)) \right \}_{ i\in\N},$$  
such that the following properties are satisfied: 
\begin{enumerate}
\item All partitions $\sigma_i$ induce the same partition of $V$.
\item There is a positive number $L_0$ such that
$$
\cal{L}_-(\sigma_i) \le L_0, \hspace{10pt} \forall i\in \mathbb{N}.   
$$
\end{enumerate}\, 
\end{theorem}

\noindent 
We refer to $\{p(n(i))\}_{i\in \mathbb{N}}$ as a {\bf diluting sequence}.    
\vspace{3pt}

The remainder of the section is organized to establish Theorem~\ref{CCOSODC}. In particular, 
we establish in Subsection~\ref{ssec:nontrivialclustering} a sufficient condition for a framework $(G,p)$ to admit a nontrivial dilute partition.  We then provide, in Subsection~\ref{ssec:ProofIII}, a proof of Theorem~\ref{CCOSODC}. 

\subsection{Existence of nontrivial dilute partitions}\label{ssec:nontrivialclustering}

Naturally, given a framework $(G,p)$,  there is a partial order defined over the set of partitions of $(G,p)$:  Let $\sigma$ and $\sigma'$ be two partitions of $(G,p)$. Let  
$V = \sqcup^m_{i = 1} V_i$ and  $V = \sqcup^{m'}_{i = 1} V'_i$ 
be the partitions of the vertex set $V$ induced by $\sigma$ and $\sigma'$, respectively. We say that $\sigma'$ is {\bf coarser} than $\sigma$, or simply write $\sigma' \prec \sigma$, if $m > m'$, and moreover, each $V_i$ is a subset of $V'_j$ for some $j = 1,\ldots, m'$.   
Recall that given a configuration $p$,  $\phi(p)$ is the diameter of $p$. Now, fix a positive number $l>0$, we establish the following result:

\begin{pro}\label{NC} Let $(G,p)$  be a framework, with $G$ a connected graph of at least two vertices and $p\in P_G$. For a positive number $l$, there exists a threshold $d>0$ such that if $\phi(p) > d$, then there is a nontrivial partition in $\Sigma(l\; p)$.
\end{pro}

\begin{proof} The proof will be carried out by contradiction: we assume that  for any $d>0$, there exists a configuration $p\in P_G$, with  $\phi(p) \ge d$, such that $\Sigma(l\; p)$ is a singleton, comprised only of the trivial partition.  

Pick any such configuration $p$, with $\phi(p)$ sufficiently large. To proceed, 
first note that there exists at least a pair of agents $x_i$ and $x_j$ in $p$, with $(v_i,v_j)\in E$, such that $\|x_j -x_i\|\le l$. This holds because otherwise,  the agent-wise partition of $(G,p)$---  
$\left\{\(\{i\}, \varnothing\), x_i\right\}^N_{i = 1}$, 
with $(\{i\}, \varnothing)$ a graph of one single vertex~$i$ and no edge---
is a nontrivial dilute partition in $\Sigma(l\; p)$.  
Now, define a partition $\sigma = \{(G_i,p_i)\}^{m}_{i=1}$ of $(G,p)$, with $G_i = (V_i,E_i)$, as follows: 
Two vertices $v_{j}$ and $v_{j'}$ are in the same subset $V_i$ if, and only if,  there is a chain of vertices $v_{j_1},\ldots,v_{j_q}$ in $V$,  with ${v_{j_1}}=v_j$ and $v_{j_q}=v_{j'}$,  such that 
\begin{equation}\label{eq:definingrule}
\(v_{j_k}, v_{j_{k+1}} \) \in E \hspace{10pt} \mbox{ and } \hspace{10pt} \| x_{j_k} - x_{j_{k+1}}\| \le l 
\end{equation}
for all $k = 1,\ldots, q-1$. 

We describe below some properties of the newly constructed partition $\sigma$. First, note that from~\eqref{eq:definingrule},  each subgraph $G_i$ is connected and $\phi(p_i)$ is bounded above; indeed, we have
\begin{equation}\label{eq:evaluatesize}
\phi(p_i)< l', \hspace{10pt} \mbox{ for } \hspace{5pt} l' := (N-1 )l.  
\end{equation} 
Furthermore, we have
\begin{equation}\label{eq:1mN}
1< m < N.
\end{equation}
To see this, first note that there exists at least an edge $(v_i,v_j)\in E$ such that $\|x_j - x_i\| \le l$, and hence there is at least a subgraph $G_k$ having more than one vertex, which implies that $m < N$. Also, note that $\phi(p)$ can be made sufficiently large; in particular, if we let $\phi(p)\ge l'$, then, from~\eqref{eq:evaluatesize}, we must have $m > 1$. 

Now, suppose that for any two adjacent frameworks $(G_i,p_i)$ and $(G_j,p_j)$, we have
$
d(p_i,p_j) > l'
$;  
then, $\sigma$ is a nontrivial partition in $\Sigma(l' \; p)$. Since $l' \ge l$, we have $\sigma \in \Sigma(l\; p)$, which is a contradiction. 
We thus assume that there are two adjacent frameworks $(G_i,p_i)$ and $(G_j,p_j)$ such that  $d(p_i,p_j) \le l'$. 
Similarly, using this condition, we define a  partition $\sigma' = \{(G'_i, \, p'_i)\}^{m'}_{i=1}$ of $(G,p)$, with $G'_i = (V'_i, E'_i)$,  
as follows: Each $V'_i$ is a union of certain subsets $V_j$, and two subsets $V_{j}$ and $V_{j'}$ are belong to the same set $V'_i$ if, and only if, there is a chain of subsets $V_{j_1},\ldots,V_{j_{q'}}$,  with $V_{j_1}=V_j$ and $V_{j_{q'}}=V_{j'}$,  such that 
$(G_{j_k}, p_{j_k})$ and  $(G_{j_{k+1}}, p_{j_{k+1}})$ are adjacent, and moreover, 
$$d(p_{j_k},p_{j_{k+1}})  \le l', \hspace{10pt}  \forall \, k = 1,\ldots, q'-1.$$
Similarly, by construction,  we have that for each $i = 1,\ldots, m'$, the subgraph $G'_i$ is connected, and  
$$ \phi(p'_i) \le l'', \hspace{10pt} \mbox{ for } \hspace{5pt} l'' :=  (2m -1) \, l'. $$ 
Furthermore, by applying the same arguments as used to prove~\eqref{eq:1mN}, we obtain  
$1 < m' < m$; indeed, we have $\sigma \succ \sigma'$.  

We then repeat the argument as above. Specifically, we assume  that there is at least a pair of adjacent frameworks $(G'_i, p'_i)$ and $(G'_j, p'_j)$ with $d(p'_i,p'_j) \le l''$. Using this  as the defining condition, we obtain another nontrivial partition $\sigma''$ of $(G,p)$, with $\sigma'' \prec \sigma' $. 
Continuing with this process, we obtain a chain of partitions of $(G,p)$ as 
$
\sigma \succ \sigma' \succ \sigma'' \succ\cdots
$. 
Since there are only finitely many partitions of $(G,p)$, the chain terminates in finite steps. For simplicity, but without loss of any generality, we assume that the chain stops at $\sigma'$. In other words, for any two adjacent frameworks $(G'_i, p'_i)$ and $(G'_j, p'_j)$, we have $d(p'_i,p'_j)> l'' $. 
But then, $\sigma'$ is in  $\Sigma(l''\; p)$. Since $l'' \ge l$,  we have $\sigma'\in \Sigma(l\; p)$,  which is a contradiction. This completes the proof. 
\end{proof}

\begin{Remark}\label{rmk:2}
Note that for any configuration $p$, we have $\phi(p) \ge d_+(p)$. Thus, from Proposition~\ref{NC}, we have that for any $l > 0$, there exists $d > 0$ such that if $d_+(p)> d$, then there is a nontrivial partition in $\Sigma(l\; p)$.
\end{Remark}

\subsection{Proof of Theorem~\ref{CCOSODC}}\label{ssec:ProofIII}
For simplicity but without loss of generality, we assume  that both sequences $\{ \phi(p(i)) \}_{i\in \N}$ and $\{l_i\}_{i\in \N}$ monotonically increase and approach to infinity. Note that in the most general case, the condition of monotonicity can be achieved by passing the original sequences to subsequences.  
The proof of Theorem~\ref{CCOSODC} is carried out by induction on the number of vertices of $G$. 

For the base case $N = 2$, we  
write $p(i) = (x_1(i),x_2(i))$; then $\phi(p(i)) = \|x_2(i) - x_1(i)\|$. 
For simplicity, we assume that 
$
\phi(p(i)) > l_i 
$ for all $ i\in\N$ (without passing to a subsequence).  
Let $\sigma_i$ be the agent-wise partition of $(G,p(i))$. Then, the following hold: (i) each $\sigma_i$ is in $\Sigma(l_i\; p(i))$; (ii) $\cal{L}_-(\sigma_i) = 0$ for all $i\in \N$; and (iii)  all the $\sigma_i$  induce the same partition of $V$, i.e., $V = \{v_1\} \cup\{v_2\}$. This establishes the base case. 

For the inductive step, we assume that Theorem~\ref{CCOSODC} holds for $N \le k-1$, and prove for $N = k$. 
Since $\{\phi(p(i))\}_{i\in\N}$ monotonically increases and approaches to infinity,  from Proposition~\ref{NC}, we have that for each $i\in \N$, there is a number
$j_i\in \mathbb{N}$ such that  if $j \ge j_i$, then there is a nontrivial partition of $(G,p(j))$ in $\Sigma(l_i, p(j))$. Without loss of generality, we assume that $j_i = i$ for all $i\in \N$. 
So, for each framework $(G,p(i))$, there is a nontrivial partition $\sigma_i$ in $\Sigma(l_i, p(i))$. 
Since there are only finitely many  partitions of~$V$, there must be a  subsequence of $\{\sigma_i\}_{i\in \N}$ such that all the partitions in the subsequence induce the same partition of $V$. Again, for simplicity but without loss of generality, we assume that the subsequence can be chosen as the original sequence $\{\sigma_i\}_{i\in \N}$ itself. 

Now, suppose that $\{\cal{L}_-(\sigma_i)\}_{i\in\N}$ is bounded; then, the sequence $\{p(i)\}_{i\in\N}$ is a diluting sequence, and hence we complete   the proof. 
We thus assume that $\{\cal{L}_-(\sigma_i)\}_{i\in \N}$ is unbounded.  
First, let 
$V = \sqcup^m_{j=1} V_j$ be the partition of $V$ induced by $\sigma_i$, for all $i\in \N$. Let $G_j$ be the subgraph induced by $V_j$, and write 
$$ \sigma_i = \{(G_j, p_j(i))\}^m_{j=1}.$$ 
Without loss of generality, we assume that  
$\{p_1(i)\}_{i\in \N}$ is unbounded. For simplicity, we assume that all the other sequences $\{p_j(i)\}_{i\in\N}$, for $j = 2,\ldots, m$, are bounded. But the arguments below can be used to prove general cases.

Since all the partitions $\sigma_i$ are nontrivial, we have that $G_1$ is a proper subgraph of $G$. 
For a framework $(G_1, p_1)$, with $p_1\in P_{G_1}$,  and a positive number $l$, let $\Sigma_1(l\; p_1)$ be the set of dilute partitions of $(G_1, p_1)$ with respect to $l$. Appealing to the induction hypothesis, we obtain a subsequence of configurations $\{p_1(n_i)\}_{i\in \N}$, with $n_i \ge i$, and a sequence of nontrivial partitions of $(G_1, p_1(n_i))$:   
$$\left \{\sigma'_i \in \Sigma_1(l_i\; p_1(n_i)) \right\}_{ i\in \N }.$$ 
The partitions above satisfy the following two conditions: 
\begin{enumerate}
\item[a)] All $\sigma'_i$ induce the same partition of $V_1$:  
$
V_1 = \sqcup^{m'}_{j = 1} V_{1_j}
$. 
\item[b)] There is a positive number $L'_0$ such that
$$
\cal{L}_-(\sigma'_i) \le L'_0 \hspace{10pt}
  \forall i\in\N.$$ 
\end{enumerate}\,

We now use $\sigma'_i$ and $\sigma_{n_i}$ to construct a new partition $\sigma_i^*$ of $(G, p(n_i))$:  
First, note that since $\{l_i\}_{i\in \N}$ monotonically increases and $n_i \ge i$ for all $i \in \N$, we have $ l_{n_i} \ge l_i$ for all $i\in\N$. 
So, if we write $$\sigma'_i = \{(G_{1_j}, p_{1_j}(n_i))\}^{m'}_{j = 1},$$ and define a partition of $(G,p(n_i))$ as follows:
$$
\sigma^*_i := \{G_{1_j}, p_{1_j}(n_i)\}^{m'}_{j = 1} \cup \{G_j, p_{j}(n_i)\}^{m}_{j = 2},
$$
then $\sigma^*_i$ is in fact an element in $\Sigma(l_i\; p(n_i))$. 
Furthermore,  from condition~b) above, we have
$$
\cal{L}_-(\sigma^*_i) \le \max \left\{ L'_0, \, \phi(p_2(n_i)),\ldots, \phi(p_m(n_i))\right\}.
$$
Since each sequence $\{\phi(p_j(n_i))\}_{i\in\N}$, for $j= 2,\ldots, m$, is by assumption bounded above, we conclude that there exists a positive number $L_0$ such that
$
\cal{L}_-(\sigma^*_i) \le L_0
$  for all  $i\in\N$. 
We have thus shown that $\{p(n_i)\}_{i\in \N}$ is a diluting sequence. \hfill{$\blacksquare$}

\section{Analysis and Proof of Theorem~\ref{thm:MAIN}}\label{sec:cgf}
This section is devoted to the proof of Theorem~\ref{thm:MAIN}.   We start with a brief outline of the proof.  In subsection~\ref{sec:BSE}, we establish the first part of Theorem~\ref{thm:MAIN}, i.e., we show that the distances between neighboring agents in an equilibrium are bounded both above and below.   
In subsection~\ref{ssec:sct}, we introduce a class of trajectories generated by system~\eqref{MODEL}, termed {\it self-clustering trajectories}. 
Roughly speaking, a {\it self-clustering trajectory} is such that the agents in the configuration evolve along time to  form disjoint clusters, with the intra- and inter-cluster distances, bounded above and below, respectively, by certain prescribed thresholds.  We show that any self-clustering trajectory remains bounded if the interactions between agents in different clusters  are all attractions. 
In subsection~\ref{ssec:sotgf}, we prove  the convergence of system~\eqref{MODEL}. The proof is carried out by contradiction: we show that if there were an unbounded trajectory generated by system~\eqref{MODEL}, then it would be a self-clustering trajectory, and moreover, the interactions between agents in different clusters are all attractions after a finite amount of time. Then, by appealing to the results derived in subsection~\ref{ssec:sct}, we conclude that any such trajectory is bounded which is a contradiction.

\subsection{Bounded sizes of equilibria}\label{sec:BSE}

In this subsection, we show that there exists positive numbers  $D_+$ and $D_-$ such that if $p$ is an equilibrium of system~\eqref{MODEL}, then
\begin{equation}\label{D-dpD+}
D_- \le d_-(p) \le d_+(p) \le D_+.
\end{equation}



\vspace{3pt}
\noindent
{\bf Existence of an upper bound}. 
Recall that $\alpha_+$ (defined in~\eqref{eq:defx+}) is such that 
\begin{equation*}
g_{ij}(d) > 0, \hspace{10pt} \forall \, d \ge \alpha_+ \mbox{ and } \forall\, (v_i,v_j)\in E;  
\end{equation*}
we then set 
\begin{equation}\label{eq:eqfirstdefD+}
D_+ := (N - 1)\,\alpha_+.
\end{equation}  
We show that if $p$ is an equilibrium of system~\eqref{MODEL}, then $d_+(p) \le D_+$. 
The proof is carried out by contradiction. 

Let $p$ be an equilibrium with $d_+(p) > (N - 1)\alpha_+$. Without loss of generality, we assume that $\|x_N - x_1\| = d_+(p)$.  Let $x^j_i$, for $j =1,\ldots, n$,  be the $j$-th coordinate of $x_i$; by rotating and/or translating $p$ if necessary, we assume that both $x_1$ and $x_N$ are on the first-coordinate, with $x^1_1 < x^1_N$. In other words, we have 
\begin{equation}\label{eq:condition1Nov1}
x^1_N - x^1_1 = d_+(p) > (N - 1) \alpha_+ .
\end{equation}
Since there are only $N$ agents,  there must be a partition 
$
V = V' \sqcup V''
$, with $v_1\in V'$ and  $v_N\in V''$, such that 
$$
x^1_{j} - x^1_{i}   > \alpha_+, \hspace{10pt} \forall \, v_i\in V' \mbox{ and } \forall\, v_j\in V''.
$$  
Indeed, if such bi-partition does not exist, then there is a chain $v_{i_1}, \ldots,v_{i_N}$, with $v_{i_1} = v_1$ and $v_{i_N} = v_N$, such that 
$
x^1_{i_{j+1}} - x^1_{i_j} \le \alpha_+
$   for all  $j = 1,\ldots, N-1$. 
But then, 
$$
x^1_N - x^1_1 = \sum^{N-1}_{j=1} \(x^1_{i_{j+1}} - x^1_{i_{j}} \)  \le (N - 1) \alpha_+,
$$
which contradicts~\eqref{eq:condition1Nov1}.  

Following the partition $V = V' \sqcup V''$, we define a subset of $E$ as follows:
$$
E^*:= \left\{  (v_i, v_j) \in E \mid v_i\in V',\, v_j\in V'' \right \},
$$
which is nonempty because $G$ is connected. We further define two variables as follows: 
$$
s'(p):= \sum_{v_i \in V'} x^1_i \hspace{10pt} \mbox{ and } \hspace{10pt} s''(p):= \sum_{v_i \in V''} x^1_i. 
$$  
The dynamics of $s'(p)$ and $s''(p)$ are given by
$$
\frac{d}{dt} s'(p) = -\frac{d}{dt} s''(p)  = \sum_{(v_i,v_j) \in E^*} g_{ij}(d_{ij}) (x^1_j - x^1_i ). 
$$
Note that for each $(v_i,v_j)\in E^*$, we have  
$
x^1_j - x^1_i > \alpha_+ 
$, and hence $g_{ij}(d_{ij}) > 0$. So,   
$$
\frac{d}{dt} s'(p) = -\frac{d}{dt} s''(p) > 0,
$$
which contradicts the fact that $p$ is an equilibrium of system~\eqref{MODEL}. We have thus shown that $D_+$, defined in~\eqref{eq:eqfirstdefD+}, is  an upper bound for $d_+(p)$ for $p$ an equilibrium. 

\vspace{3pt}
\noindent
{\bf Existence of a lower bound}.  
We first have some notations. Let $(v_i,v_j)\in E$ be an edge of $G$; for any positive number $d$,  we define a function in $\operatorname{C}(\R_+, \R)$ as follows:
\begin{equation}\label{eq:defbargij}
\ol g_{ij}(d) := dg_{ij}(d).
\end{equation}
Let  $S$ be a subset of $\R$, and let $\ol g^{-1}_{ij}(S)$ be a subset of $\R_+$ defined by 
$$
\ol g_{ij}^{-1}(S) := \left\{ d \in \R_+ \mid \ol g_{ij}(d) \in S  \right\}.
$$
With the notations above, we establish the following fact:

\begin{lem}\label{lem:lem2}
Let $g_{ij}$, for $(v_i,v_j) \in E$, be an attraction/repulsion function. Then,  
$$
\lim_{\eta \to \infty} \sup \left \{d \in  \ol g^{-1}_{ij}\(\pm\, \eta \)  \right \} = 0.
$$ \,
\end{lem}

\begin{proof}
It suffices to show that for any $d\in \R_+$, there exists a positive number $\eta_d$ such that if $\eta > \eta_d$, then $\ol g^{-1}_{ij}\(\pm \, \eta\)$ is a nonempty subset of the open interval $(0,d)$.  
First, from the condition of {\it fading attraction}, we have 
$$
 \sup \left \{|\ol g_{ij}(d')| \mid d' \ge d  \right \} < \infty. 
$$
We can thus define 
$$
\eta_d :=  \sup \left \{|\ol g_{ij}(d')| \mid d' \ge  d\right \} + 1.
$$
Then,  by the fact that $\lim_{d \to 0+} \ol g_{ij}(d) = -\infty$,   
we conclude that if $\eta > \eta_d$, then $\ol g^{-1}_{ij}\(\pm \, \eta \)$ is nonempty, and is contained in $(0,d)$.  
\end{proof}

Let $d$ be a positive number;  we define a subset of $P_G$ as follows: 
\begin{equation}\label{eq:defZGd}
Z_G(d) := \left \{ p\in P_G\mid d_-(p) = d \right \}.
\end{equation}
Recall that $f(p)$ is the vector field of system~\eqref{MODEL} at $p$. With Lemma~\ref{lem:lem2}, we establish the following fact:

\begin{pro}\label{pro:nmvecfld}
Let $Z_G(d)$ be defined in~\eqref{eq:defZGd}. Then, 
$$
 \lim_{d\to 0+} \inf \left \{ \|f(p)\| \mid p\in Z_G(d) \right \} = \infty. 
$$\, 
\end{pro}

\begin{proof}
The proof will be carried out by induction on the number of vertices of $G$. For the base case $N = 2$, we have 
$$
Z_G(d) = \left \{ (x_1, x_2) \in P_G\mid \|x_2 - x_1\| = d \right \}.
$$
We also have for any $p\in P_G$, 
$$\|f(p)\| = \sqrt{2}\, \|f_1(p)\| = \sqrt{2} \, \|f_2(p)\| = \sqrt{2}\, | \ol g_{12}(d_{12}) |.$$
From the condition of {\it strong repulsion}, we have
$
\lim_{d\to 0+}\sqrt{2}\, | \ol g_{12}(d)| = \infty
$, which establishes the base case. 

For the inductive step, we assume that Proposition~\ref{pro:nmvecfld} holds for $N \le k-1$, and prove for $N = k$. The proof will be carried out by contradiction: we assume that there exists a number $\eta > 0$ such that for any $d > 0$, there is a number $d_* \in (0, d)$ and a configuration $p\in Z_{G}\(d_*\)$ such that $\|f(p)\| \le \eta$.

Choose $d$, and hence $d^*$, arbitrarily small; let $p\in Z_G(d_*)$ be such that $ \| f(p) \| \le \eta$. Without loss of generality, we assume that $(v_1,v_2)$ is an edge of $G$, and moreover,   
$ d_{12} = d_*$.    
 Since $G$ is connected, there is a connected subgraph $G' = (V', E')$ of $G$ which has $(k - 1)$ vertices and contains the edge $(v_1,v_2)$. Label the vertices of $G$ such that $V' = \{v_1,\ldots,v_{k-1}\}$.  
Note that if we let $p'$ be the sub-configuration of $p$ associated with $G'$, then $p' \in Z_{G'}(d_*)$.  Let $S'$ be the sub-system of~\eqref{MODEL} induced by $G'$, and $f'(p')$ be the associated vector field at $p'$. From the induction hypothesis,  if we let $\omega \in \R_+$ be such that 
\begin{equation}\label{eq:defgamma}
 \| f'(p') \| = \omega\,  \eta; 
\end{equation}
then,   $\omega$ can be made arbitrarily large by assuming that $d_*$ sufficiently small.

For each $v_i \in V'$, let $f'_i(p')$ be defined by restricting $f'(p')$ to~$x_i$. From~\eqref{eq:defgamma}, there exists at least a vertex $v_i\in V'$ such that 
$$  \|f'_i(p') \| \ge \omega' \,   \eta, \hspace{10pt} \mbox{for } \omega' := \omega/ \sqrt{k - 1}.$$
Note that $\omega'$ can be made arbitrarily large by increasing $\omega$. 
Without loss of generality, we assume that $p$ is rotated in a way such that
\begin{equation}\label{eq:assumptiononf'_i}
f'_i(p) = \(\| f'_i(p) \|, 0, \ldots, 0\) \in \R^n.
\end{equation}
There are two cases: 

{\it Case I}. Suppose that $(v_i,v_k)\notin E$; then,  $f_i(p) = f'_i(p)$. In particular, if we let $\omega' > 1$, then 
$$
\|f(p)\| \ge \|f_i(p)\| = \omega'\, \eta > \eta
$$ 
which is a contradiction. The proof is then complete. 

{\it Case II}. We assume that $(v_i,v_k)\in E$. 
Recall that $x^1_i$ is the first coordinate of $x_i$. Following~\eqref{eq:assumptiononf'_i}, the dynamics of $x^1_i$, in system~\eqref{MODEL}, is given by  
$$ \dot x^1_i =\| f'_i(p) \| + g_{ik}(d_{ik}) (x^1_k - x^1_i ).$$
Since $|\dot x^1_i | \le  \|f(p)\|$,  we have
\begin{equation*}
g_{ik}(d_{ik}) (x^1_k - x^1_i ) \le  \| f(p) \| - \| f'_i(p) \|.
\end{equation*}
Using the fact that $\|f(p)\|\le \eta$ and $\|f'_i(p) \|\ge \omega'\, \eta$, we obtain
\begin{equation}\label{eq:ineqforgik}
g_{ik}(d_{ik}) (x^1_k - x^1_i )  \le - (\omega' -1)\, \eta < 0,  
\end{equation}
which further implies that  
$ |\ol g_{ik}(d_{ik}) |  \ge (\omega' - 1)\,\eta$. 
Then, from Lemma~\ref{lem:lem2},  we have that $d_{ik}$ can be made arbitrarily small by increasing $\omega'$.  
In particular, if $d_{ik}$ is sufficiently small such that $g_{ik}(d_{ik}) < 0$, then, from~\eqref{eq:ineqforgik}, we have $x^1_k > x^1_i$. 

We next consider the dynamics of $x^1_k$ in system~\eqref{MODEL}: 
$$
\dot x^1_k  = g_{ik}(d_{ik})( x^1_i - x^1_k ) + \sum_{v_j\in V_k - \{v_i\}} g_{jk}(d_{jk})(x^1_j - x^1_k ).
$$ 
Combining~\eqref{eq:ineqforgik} with the fact that $|\dot x^1_k |  \le \eta$, we know that there is at least a vertex $v_j \in V_k - \{v_i\}$ such that 
\begin{equation}\label{eq:3:16pmNov1}
g_{jk}(d_{jk})( x^1_j - x^1_k )  \le -\omega'' \, \eta, \hspace{10pt} \mbox{for } \omega'' := \frac{\omega' - 2 } {k - 1}.
\end{equation}
The right hand side of the equation can be made negative by increasing $\omega'$, and hence $\omega''$. 
Appealing again to Lemma~\ref{lem:lem2}, we know that by increasing $\omega''$, we can make $d_{jk}$ arbitrarily small such that $g_{jk}(d_{jk}) < 0$. It then follows from~\eqref{eq:3:16pmNov1}  that    
$
x^1_j >  x^1_k$. 

Now, for the dynamics of $x^1_j$, we can apply the arguments above as to the dynamics of $x^1_k$. By doing so, we obtain another vertex $v_{j'}\in V_j$ such that $x^1_{j'} > x^1_{j}$. 
Furthermore, by repeating using the arguments, we obtain an infinite sequence as follows:
$$ x^1_i < x^1_k < x^1_j < x^1_{j'} < x^1_{j''} < \cdots. $$
This contradicts the fact that  $G$ has only $k$ vertices, which completes the proof.   
\end{proof}

The existence of a lower bound $D_-$ then directly follows from Proposition~\ref{pro:nmvecfld}; indeed, from Proposition~\ref{pro:nmvecfld}, we can choose $D_-$ to be such that if 
$d\le D_-$ and $p\in Z_G(d)$, then $\|f(p)\| > 1$. 
We have thus established the first part of Theorem~\ref{thm:MAIN}.

\subsection{Self-clustering trajectories}\label{ssec:sct}
We introduce in this subsection {\it self-clustering trajectories} of system~\eqref{MODEL}. Let $(G,p)$ be a framework, and $\sigma = \{(G_i,p_i)\}^m_{i=1}$ be a partition of $(G,p)$. Recall that the intra- and inter-cluster distances of the partition $\sigma$ are defined (in~\eqref{eq:intradistance} and~\eqref{eq:interdistance}, respectively) as follows:
\begin{equation*}
\left\{
\begin{array}{l}
\cal{L}_-(\sigma)= \max\left \{ \phi(p_i) \mid 1\le i \le m  \right\}, \\
\cal{L}_+(\sigma) = \min_{(i,j)} \left \{ d(p_i,p_j) \right \}
\end{array}
\right. 
\end{equation*} 
where the minimum is taken over pairs $(i,j)$, for $(G_i,p_i)$ and $(G_j,p_j)$ adjacent.  
We then have the following definition:

\begin{Definition}\label{def:selfclustering}
Let $l_0$ and $l_1$ be positive numbers. A trajectory $p(t)$ of system~\eqref{MODEL} is {\bf self-clustering}, with respect to $(l_0,l_1)$, if there exists a nontrivial partition $V = \sqcup^{m}_{i =1} V_i$ such that the following condition is satisfied: Let $G_i$ be the subgraph induced by $V_i$, and let $\sigma_t = \{(G_i, p_i(t))\}^m_{i = 1}$ be a partition of $(G,p(t))$.  Then, there exists an instant $t_0\ge 0$ such that for all $t \ge t_0$, we have
\begin{equation}\label{eq:defselfclustering}
\cal{L}_-(\sigma_t) < l_0 \hspace{10pt} \mbox{ and } \hspace{10pt} \cal{L}_+(\sigma_t) > l_1. 
\end{equation} \,

\end{Definition} 

Recall that  the number $\alpha_+$ (defined in~\eqref{eq:defx+}) is chosen such that 
$
g_{ij}(d) > 0   
$ for all $d \ge \alpha_+$  and for all $(v_i,v_j)\in E$. 
We prove in this subsection that if a trajectory $p(t)$ is self-clustering, with inter-cluster distances sufficiently large (greater than $\alpha_+$); then, $p(t)$ remains bounded along time~$t$. Precisely, we have  the following fact:

\begin{pro}\label{pro:selfclustering}
Let $l_0$ and $l_1$ be positive numbers, and assume that $l_1 > \alpha_+$. Suppose that $p(t)$ is a self-clustering trajectory with respect to $(l_0, l_1)$; then, $p(t)$ remains bounded along the evolution, i.e., 
$$
\sup\{\phi(p(t)) \mid t \ge 0\} < \infty.
$$ \,   
\end{pro}

To prove Proposition~\ref{pro:selfclustering}, we first introduce some notations. 
Let the centroid of a configuration $p(t)$ be defined as 
$$
c(p(t)) := \sum_{v_i \in V} x_i(t)/ |V|.  
$$ 
Then, by computation, we have $d c(p(t))/dt = 0$. So, for simplicity but without loss of generality, we can assume that 
$c(p(t)) = 0$  for all $t\ge 0$. 
Now, for each $t\ge 0$, let 
\begin{equation}\label{eq:defsigmatforproposition5}
\sigma_t = \{(G_i, p_i(t))\}^m_{i = 1}, \hspace{10pt} \mbox{ with } \hspace{5pt} G_i = (V_i, E_i), 
\end{equation}  
be a nontrivial partition of $(G, p(t))$. Because $p(t)$ is a self-clustering trajectory. So, we can assume that~\eqref{eq:defselfclustering} holds for the partitions $\sigma_t$ defined in~\eqref{eq:defsigmatforproposition5} for all $t \ge t_0$. Further, for simplicity, we assume that $t_0 = 0$, i.e.,~\eqref{eq:defselfclustering} holds since the starting time.   
Let $\cal{I} := \{1,\ldots, m\}$ be the index set for the frameworks $\{(G_i,p_i(t))\}^m_{i=1}$ associated with $\sigma_t$. For a subset $\cal{I}'$ of $\cal{I}$, let $$G_{\cal{I}'} = (V_{\cal{I}'}, E_{\cal{I}'}), \hspace{10pt} \mbox{ with } \hspace{5pt} V_{\cal{I}'}: =\sqcup_{j\in \cal{I}'} V_j $$  
be the subgraph of $G$ induced by $V_{\cal{I}'}$, and let  
$
p_{\cal{I}'}(t)  
$ 
be the sub-configuration of $p(t)$ associated with $G_{\cal{I}'}$. Similarly, let the centroid of $p_{\cal{I}'}(t)$ be 
$c(p_{\cal{I}'}(t)) := \sum_{v_{i}\in V_{\cal{I}'}} x_i(t) /|V_{\cal{I}'}|$.

Next, we introduce a set of time-dependent variables, encoding certain metric properties of $p(t)$ along time~$t$.  
First, for a subset $\cal{I}'\subset\cal{I}$, let a continuously differentiable function in~$t$ be defined as 
$
\pi(\cal{I}'\; t):= \| c(p_{\cal{I}'}(t)) \|
$. 
Then, for an integer $k = 1,\ldots, m$,  we define a continuous function by 
\begin{equation}\label{eq:defPIkt}
\Pi(k\; t) := \max_{\cal{I}'}\{\pi(\cal{I}'\; t) \mid  |\cal{I}'| = k\}. 
\end{equation}
For example, if $k = 1$, then
\begin{equation*}\label{eq:Pi1thehe}
\Pi(1\; t) = \max\{\|c(p_i(t))\|  \mid i = 1,\ldots, m\}; 
\end{equation*}
and if $k = m$, then   
$$
\Pi(m\; t) = \|c(p(t))\| = 0.
$$
We  note here, without a proof,  the following fact that for any $t\ge 0$,  
\begin{equation*}\label{eq:1:30pm}
\Pi(1\; t) \ge \ldots \ge \Pi(m\; t) = 0.
\end{equation*}

Now, fix an integer $k = 1,\ldots, m-1$,  we relate below $\Pi(k,t)$ and $\Pi(k+1,t)$ by formalizing the following fact: 
if $\Pi(k\; t)$ is expanding at a certain instant~$t$, then $\Pi(k+1\; t)$ cannot be too small. Precisely, we have the following result:

\begin{lem}\label{lem:boundedforPI}
Let $p(t)$ be the self-clustering trajectory, with respect to $(l_0,l_1)$, in Proposition~\ref{pro:selfclustering}.  
Fix an instant~$t > 0$, and let $r>0$ be such that   $\Pi(1\; t') \le r$ for all $t' \le t$.   
Suppose that there is an integer $k= 1,\ldots, m-1$ such that $\Pi(k\; t') \le \Pi(k\; t)$  for all $t' \le t$;  
then,  
$$
\Pi(k+1\; t) \ge r - N (r - \Pi(k\; t) ) -  2 l_0. 
$$ \,
\end{lem}

We refer to Appendix~A for a proof of Lemma~\ref{lem:boundedforPI}. 
 With Lemma~\ref{lem:boundedforPI}, we prove Proposition~\ref{pro:selfclustering}. 

\begin{proof}[Proof of Proposition~\ref{pro:selfclustering}]
Let $\sigma_t$ be defined in~\eqref{eq:defsigmatforproposition5} as the nontrivial partitions associated with the self-clustering trajectory $p(t)$. Let $\Pi(k, t)$, for $k = 1,\ldots, m$, be defined in~\eqref{eq:defPIkt}. We first show that  
\begin{equation}\label{eq:phiptcaonima}
 \phi(p(t)) <2(\Pi(1\; t) + l_0), \hspace{10pt} \forall\, t \ge 0.
\end{equation} 
Let $v_{i}$, $v_{j}$ be any two vertices in $V$; we assume that $v_{i} \in V_{i'}$ and $v_{j}\in V_{j'}$. Then, by the triangle inequalities, the distance $d_{ij}(t)$ between $x_i(t)$ and $x_j(t)$ is bounded above by the sum of three terms: 
$$
\begin{array}{lll}
d_{ij}(t)  & \le & \|x_{i}(t) - c(p_{i'}(t))\|  + \| c(p_{j'}(t)) - x_{j}(t) \| \\
& & +   \|  c(p_{i'}(t)) - c(p_{j'}(t)) \|;
\end{array}
$$
for the first two terms, we have
$$
\left\{
\begin{array}{l}
 \|x_{i}(t) - c(p_{i'}(t))\| < \phi(p_{i'}(t)) < l_0, \\
 \| c(p_{j'}(t)) - x_{j}(t) \| < \phi(p_{j'}(t)) < l_0; 
\end{array}
\right.
$$
for the last term, we have
$$
\| c(p_{i'}(t)) - c(p_{j'}(t))\| \le   2 \Pi(1\; t).
$$
It then follows that $d_{ij}(t) < 2(\Pi(1\; t) + l_0)$, and hence~\eqref{eq:phiptcaonima} holds. 

It thus suffices to show that $\sup \{\Pi(1, t)\mid t \ge 0 \} < \infty$. The proof is carried out by contradiction. Suppose that, to the contrary, for any $r \ge 0$, there exists an instant $t_1$ such that $\Pi(1,t_1) = r$. Choose $r$ sufficiently large such that $r > \Pi(1\; 0)$, and let $t_1$ be such that 
$$\Pi(1\; t)  \le \Pi(1\; t_1) = r, \hspace{10pt} \forall\,  t \le t_1.$$  
Then, from Lemma~\ref{lem:boundedforPI}, we have 
$$
\Pi(2\; t_1) \ge r -  2 l_0. 
$$
We may increase $r$, if necessary, so that $r - 2 l_0 > \Pi(2\; 0)$. Choose an instant~$t_2 \in (0, t_1]$ such that     
$$
\Pi(2\; t)  \le \Pi(2\; t_2) = r -  2 l_0, \hspace{10pt} \forall\,  t \le t_2.
$$
Then, appealing again to Lemma~\ref{lem:boundedforPI}, we obtain
$$
\Pi(3\; t_2) \ge r -  2(N + 1) l_0.
$$
Repeating this argument, we then obtain a time sequence  
$
t_1 \ge  \ldots \ge t_{m-1}
$ 
such that for all $k = 1,\ldots, m-1$, we have
$$
\Pi(k+1\; t_k) \ge r - 2 \sum^{k-1}_{i = 0} N^i  l_0.
$$
In particular, for $k= m$, we have
\begin{equation}\label{eq:8:24pm}
0 = \Pi(m\; t_{m-1}) \ge r - 2 \sum^{m - 2}_{i = 0} N^i  l_0,
\end{equation}
which is a contradiction because  $r$ can be chosen arbitrarily large, and hence the right hand side of~\eqref{eq:8:24pm} is positive. This completes the proof.
\end{proof}

\subsection{Convergence of the gradient flow}\label{ssec:sotgf}
We now return to the proof of Theorem~\ref{thm:MAIN}. We show that for any initial condition $p(0)$ in $P_G$, the trajectory $p(t)$ converges to the set of equilibria. 
The proof will be carried out by contradiction, i.e., we assume that there is an initial condition $q(0)\in P_G$ such that the trajectory $q(t)$ of system~\eqref{MODEL} is unbounded. In the remainder of the section, we fix the trajectory $q(t)$, and derive contradictions.  


Since $q(t)$ is unbounded, there is a time sequence $\{t_i\}_{i\in N}$, with $\lim_{i\to \infty} t_i = \infty$, such that $\{q(t_i)\}_{i\in\N}$ is unbounded. 
Choose a monotonically increasing sequence $\{l_i\}_{i\in \N}$ out of $\R_+$, and let $\lim_{i\to \infty} l_i = \infty$.  
From Theorem~\ref{CCOSODC}, there is a diluting sequence  $\{q(t_{n_i})\}_{i\in \N}$, as a subsequence of $\{q(t_i)\}_{i\in \N}$, together with a sequence of nontrivial partitions $\{\sigma_{t_i}\}_{i\in\N}$ satisfying the following properties:
\begin{enumerate}
\item All partitions $\sigma_{t_i}$, for $i\in \N$, induce the same partition of $V$: 
$
V = \sqcup^m_{i=1} V_i
$.
\item There exists $L_0 > 0$ such that
$
\cal{L}_-(\sigma_i) \le L_0
$ for all $ i\in \N$.
\end{enumerate}  \,
Without loss of generality, we assume that $n_i = i$ for all $i\in\N$, i.e., the subsequence $\{q(t_{n_i})\}_{i\in\N}$ can be chosen as $\{q(t_i)\}_{i\in\N}$ itself. We can also assume that  $L_0$ is large enough so that $L_0 \ge D_+$, with $D_+ = (N - 1) \, \alpha_+$ defined in~\eqref{eq:eqfirstdefD+}.  

Following the partition $V = \sqcup^m_{i=1} V_i$, we let $G_i = (V_i, E_i)$ be the subgraph induced by $V_i$. For each framework $(G,q(t))$, we let $\sigma_t$ be the nontrivial partition of $(G,q(t))$ defined as
\begin{equation}\label{eq:defsigmat}
\sigma_t:= \{(G_i, q_i(t))\}^m_{i = 1};  
\end{equation}
note that for each $i \in \N$, we have  $\sigma_{t_i} \in \Sigma(l_i\; q(t_i))$. 
We show below that if $q(t)$ is unbounded, then it {\it has to} be a {\it self-clustering} trajectory with respect to $(L_0,l_i)$ for all $i\in \N$. 
Precisely, we state the following result:

\begin{pro}\label{pro:relatepartitiontodynamicalsystem}
Suppose that $q(t)$ were an unbounded trajectory generated by system~\eqref{MODEL};   
let $\sigma_{t}$, for $t\ge 0$, be the partition of $(G,q(t))$ defined in~\eqref{eq:defsigmat}. 
Then, for each $i\in \N$, there would be a $j_i\in \N$ such that for all $t \ge t_{j_i}$, we have
\begin{equation*}\label{eq:contradictioncondition1}
\sigma_t\in \Sigma(l_i\; q(t)) \hspace{5pt} \mbox{ and } \hspace{5pt}  \cal{L}_-(\sigma_t) <  L_0. 
\end{equation*}
In particular, $q(t)$ would be a self-clustering trajectory with respect to $(L_0, l_i)$ for all $i\in \N$. 
\end{pro}

  
We refer to Appendix~B for a proof of Proposition~\ref{pro:relatepartitiontodynamicalsystem}. With Propositions~\ref{pro:selfclustering} and~\ref{pro:relatepartitiontodynamicalsystem}, we prove the second part of Theorem~\ref{thm:MAIN}.

\begin{proof}[Proof of the second part of Theorem~\ref{thm:MAIN}]
The proof is carried out by contradiction; we assume that there exists an unbounded trajectory $q(t)$ of system~\eqref{MODEL}.  But then, 
by combining Propositions~\ref{pro:selfclustering} and~\ref{pro:relatepartitiontodynamicalsystem}, we derive a contradiction:  First, choose an~$i\in \N$ such that $l_i > \alpha_+$, then from Proposition~\ref{pro:relatepartitiontodynamicalsystem}, $q(t)$ is a self-clustering trajectory with respect to $(L_0, l_i)$; On the other hand, from Proposition~\ref{pro:selfclustering}, we have
$
\sup_{t \ge 0} \phi(q(t))  < \infty  
$,  
which contradicts the assumption that $q(t)$ is unbounded. 
We thus conclude that for any initial condition $p(0)\in P_G$, the trajectory $p(t)$ is bounded, and hence converges to the set of equilibria. This completes the proof. 
\end{proof}





\section{Conclusions}
We have established in this paper the convergence of the multi-agent system~\eqref{MODEL} under the assumption that the interaction functions $g_{ij}$, for $(v_i,v_j) \in E$, have fading attractions. 
To tackle this propblem, we introduced dilute partitions in section~III, as a new tool, to characterize the behaviors of trajectories generated by system~\eqref{MODEL}.  The use of dilute partitions  enabled us to grasp the qualitative properties of the dynamics of the formations that are needed to prove the convergence results:  On one hand, it reveals the fact that self-clustering trajectories are all bounded, as shown in Proposition~\ref{pro:selfclustering}. On the other hand, it precludes the possibility of system~\eqref{MODEL} having unbounded trajectories, as implied by Proposition~\ref{pro:relatepartitiontodynamicalsystem}.  
Further, we note that the class of dilute partitions is itself a rich question. We have exhibited in section~III some intriguing facts of it: For example, the existence of a nontrivial dilute partition in Proposition~\ref{NC}, and the existence of diluting sequence in Theorem~\ref{CCOSODC}.   These facts are independent of the dynamical system~\eqref{MODEL}, and hence can be used to solve other difficult multi-agent control problems that involve large sized formations.  Future work may focus on establishing system convergence under the assumption that interaction functions $g_{ij}$'s have not only fading attractions, but also {\em finite repulsions}. We are also interested in studying the system behavior when the network topology $G$ is directed and/or time-varying. Of course, in either of the two cases above, system~\eqref{MODEL} is not a gradient system anymore. It is thus interesting to know whether or not trajectories of system~\eqref{MODEL} still converge. This also lies in the scope of our future research.

\section*{Acknowledgements}
The author here is grateful for discussions with Prof. Roger Brockett at Harvard University, Prof. Tamer Ba\c sar and Prof. M.-A. Belabbas at the University of Illinois at Urbana-Champaign.  






\bibliographystyle{unsrt}
\bibliography{FC}

\section*{Appendix A}
\begin{proof}[Proof of Lemma~\ref{lem:boundedforPI}] 
Let $\cal{I}' \subset \cal{I}$, with $|\cal{I}'| = k$, be chosen such that  
$
\pi(\cal{I}'\; t) = \Pi(k\; t)  
$. 
Let $\langle\cdot, \cdot \rangle $ be the standard inner-product in $\R^n$. For a vector $v\in \R^n$, let 
$$
\hat v :=
\left\{
\begin{array}{ll}
v /\|v\| & \mbox{if } v \neq 0, \\
0  & \mbox{otherwise}.
\end{array}
\right. 
$$  
We first show that  for all $i\in \cal{I}'$, 
\begin{equation}\label{eq:idon'tsleepwelltoday}
\langle \hat c(p_{\cal{I}'}(t)), c(p_{i}(t)) \rangle \ge r - N(r - \Pi(k\; t)).
\end{equation}
Let $w_i: = |V_i| / |V_{\cal{I}'}|$; then, we can express $\pi(\cal{I}'\; t)$ as 
$$
\pi(\cal{I}'\; t) =\langle \hat c(p_{\cal{I}'}(t)),\, \sum_{i\in \cal{I}'}w_i \, c(p_i(t)) \rangle. 
$$
Then, using the fact that for all $i\in \cal{I}' $, 
$$\langle \hat c(p_{\cal{I}'}(t)), c(p_j(t)) \rangle \le \|c(p_j(t))\| \le r,$$  
we obtain
$$
 \langle \hat c(p_{\cal{I}'}(t)), c(p_i(t)) \rangle \ge  r - \frac{1}{w_i} \(r - \pi(\cal{I}'\; t)\).  
$$
Since $1/w_i \le N$, we establish~\eqref{eq:idon'tsleepwelltoday}.  

To proceed, we consider the time derivative of $\pi(\cal{I}'\; t)^2$ at~$t$:  First, let  
$$
E_{\cal{I}'}:=\{(v_{a},v_{b}) \in E \mid v_{a}\in V_{\cal{I}'},\, v_{b} \notin V_{\cal{I}'} \}. 
$$ 
Note that $E_{\cal{I}'}$ is nonempty because (i) $V_{\cal{I}'}$ is a proper subset of $V$since $\cal{I}'$ is a proper subset of $\cal{I}$, and (ii) $G$ is connected. Further, for an edge $(v_a, v_b)\in E_{\cal{I}'}$, let 
$$\rho_{ab}(t) := \left\langle c(p_{\cal{I}'}(t)) , x_{b}(t) - x_{a}(t) \right\rangle.$$ 
Then, with the definitions of $E_{\cal{I}'}$ and $\rho_{ab}(t)$, we have 
\begin{equation}\label{eq:lajitong}
\frac{d}{dt}\pi(\cal{I}'\; t)^2 = 2\sum_{(v_{a},v_{b}) \in E_{\cal{I}'} } g_{ab}(d_{ab}(t))\, \rho_{ab}(t).
\end{equation} 
Note that $d\pi(\cal{I}'\; t)^2/dt \ge 0$ because $\Pi(k\; t') \le \Pi(k\; t) $ for all $t' \le t$.  We also note that $g_{ab}(d_{ab}(t)) > 0$ for all $(v_a,v_b) \in E_{\cal{I}'}$. This holds because $p(t)$ is a self-clustering trajectory with respect to $(l_0, l_1)$ and $l_1 > \alpha_+$, which in particular implies that $d_{ab}(t) > \alpha_+$. 
All then imply that there is at least an edge $(v_a,v_b) \in E_{\cal{I}'}$ such that 
$\rho_{ab}(t) \ge 0$.

Choose any such edge $(v_a,v_b) \in E_{\cal{I}'}$, and let indices~$i, j\in \cal{I}$ be such that $v_a\in V_{i}$ and $v_{b} \in V_{j}$. 
It should be clear that  $i\in \cal{I}'$ and $j\notin \cal{I}'$. Note that since $\cal{I}'$ is chosen such that $\Pi(k\; t) = \pi(\cal{I}'\; t)$, we have 
$$\langle \hat c(p_{\cal{I}'}(t)), c(p_j(t)) \rangle \le \| c(p_{\cal{I}'}(t))\|.$$ Because otherwise, we can first find an $i'\in \cal{I}'$ with $\langle \hat c(p_{\cal{I}'}(t)), c(p_i'(t))\rangle \le \|c(p_{\cal{I}'}(t))\|$, and then replace this $i'\in \cal{I}'$ with~$j$. By doing so, we obtain a strictly larger $\pi(\cal{I}'\;t)$, which is a contradiction. On the other hand,     
we show that 
\begin{equation}\label{eq:12:22pm}
 \left\langle  \hat c(p_{\cal{I}'}(t)),\, c(p_{j}(t)) \right\rangle > \langle \hat c(p_{\cal{I}'}(t)), c(p_{i}(t)) \rangle - 2l_0. 
\end{equation}
To prove~\eqref{eq:12:22pm}, first note that 
$$
\left\{
\begin{array}{l}
\| x_{a}(t) - c(p_{i}(t))\| < \phi(p_{i}(t)) < l_0 \vspace{3pt}\\
\| x_{b}(t) - c(p_{j}(t))\| < \phi(p_{j}(t))  < l_0.
\end{array}
\right.
$$
So, we obtain
\begin{equation*}\label{eq:1:59pm}
\left\{
\begin{array}{l}
\langle  \hat c(p_{\cal{I}'}(t)),    c(p_{i}(t)) - x_{a}(t)  \rangle < l_0  \vspace{3pt}\\
\langle  \hat c(p_{\cal{I}'}(t)),   x_{b}(t) - c(p_{j}(t)) \rangle < l_0,
\end{array}
\right. 
\end{equation*}
which implies that 
$$
\left\langle  \hat c(p_{\cal{I}'}(t)),\, c(p_{j}(t)) \right\rangle > \langle \hat  c(p_{\cal{I}'}(t)), c(p_{i}(t)) \rangle - 2l_0 + \rho_{ab}(t).
$$
Since $\rho_{ab}(t) \ge 0$, we establish~\eqref{eq:12:22pm}.

Now, by combining~\eqref{eq:idon'tsleepwelltoday} and~\eqref{eq:12:22pm}, we obtain the following inequality:
\begin{equation}\label{eq:7:27pmatcaffebene}
 \left\langle  \hat c(p_{\cal{I}'}(t)),\, c(p_{j}(t)) \right\rangle \ge r - N (r - \Pi(k\; t) ) -  2 l_0.
\end{equation}
Let $\cal{I}'':= \cal{I}' \sqcup \{j\}$. Since $j\notin \cal{I}'$, we have $|\cal{I}''| = k+1$. It now suffices to show that 
$$
\pi(\cal{I}''\; t) \ge  \left\langle  \hat c(p_{\cal{I}'}(t)),\, c(p_{j}(t)) \right\rangle. 
$$
Let 
$\tilde w_{\cal{I}'}:= |V_{\cal{I}'}|/|V_{\cal{I}''}|$ and  $\tilde w_j :=   |V_j| / |V_{\cal{I}''}|$. It should be clear that $\tilde w_{\cal{I}'} + \tilde w_{j} = 1$, and $c(p_{\cal{I}''}(t)) =  \tilde w_{\cal{I}'} c(p_{\cal{I}'}(t)) + \tilde w_{j} c(p_j(t))$. We now express $\pi(\cal{I}''\; t)$ as 
$$\pi(\cal{I}''\; t) =  \tilde w_{j} \langle \hat c(p_{\cal{I}''}(t)),c(p_j(t))\rangle +  \tilde w_{\cal{I}'} \langle \hat c(p_{\cal{I}''}(t)),c(p_{\cal{I}'}(t)) \rangle.$$
For the first inner-product, we have  
$$
\langle \hat c(p_{\cal{I}''}(t)), c(p_j(t)) \rangle \ge \langle \hat c(p_{\cal{I}'}(t)),\, c(p_{j}(t)) \rangle,
$$
and the equality holds if and only if $c(p_{\cal{I}'}(t))$ and $c(p_j(t))$ are aligned. For the second inner-product, first note that 
$$ \|c(p_{\cal{I}''}(t))\| \le \Pi(k+1\; t) \le \Pi(k\; t) = \|c(p_{\cal{I}'}(t))\|,$$
and hence
$$
\langle \hat c(p_{\cal{I}''}(t)), c(p_{\cal{I}'}(t)) \rangle \ge\langle \hat c(p_{\cal{I}'}(t)), c(p_{\cal{I}''}(t)) \rangle.
$$
Then, using the fact that $$\langle \hat c(p_{\cal{I}'}(t)), c(p_j(t)) \rangle \le \| c(p_{\cal{I}'}(t))\|,$$ we obtain 
$$
\langle \hat c(p_{\cal{I}'}(t)), c(p_{\cal{I}''}(t)) \rangle \ge  \langle \hat c(p_{\cal{I}'}(t)),\, c(p_{j}(t)) \rangle.
$$
Combining the facts above, we conclude that $\pi(\cal{I}''\; t) \ge  \left\langle  \hat c(p_{\cal{I}'}(t)),\, c(p_{j}(t)) \right\rangle$,  
 which completes the proof. 
\end{proof}

\section*{Appendix B}
We establish here Proposition~\ref{pro:relatepartitiontodynamicalsystem}. We need to first introduce a class of subsets of $P_G$, termed {\it dissipation zones}, and establish properties  that are needed to prove Proposition~\ref{pro:relatepartitiontodynamicalsystem}. 

\subsection{Dissipation zones}\label{ssec:dz}
Let $d$ be a positive number. For each $(v_i,v_j) \in E$, 
define a subset $X_{G,ij}(d)$ of $P_G$ as follows:
\begin{equation}\label{eq:defXGijd}
X_{G,ij}(d) :=\left  \{ p\in P_G  \mid  \|x_j - x_i\| = d \right \}; 
\end{equation}
We further define
$X_{G}(d):= \cup_{(v_i,v_j) \in E}\,  X_{G,ij}(d)$. 
Note that if $d > D_+$, then $X_{G}(d)$ does not contain any equilibrium of system~\eqref{MODEL}. We call any such set $X_{G}(d)$, for $d > D_+$, a {\bf dissipation zone}. 
Define a function $\mu:\R_+\longrightarrow \R$ as follows:
\begin{equation}\label{eq:defmud}
\mu(d) := \inf\{\|f(p)\| \mid p\in X_{G}(d) \};
\end{equation}
we establish in this subsection the following fact:

\begin{pro}\label{pro:defnud}
Let $\mu:\R_+ \longrightarrow \R$ be defined in~\eqref{eq:defmud}. Then, 
\begin{enumerate}
\item $\mu$ is continuous. 
\item $\mu(d) > 0$ for all $d \ge D_+$.  
\end{enumerate}\,
\end{pro}


Recall that $d_-(p)$ (defined in~\eqref{eq:defd-d+}) is the minimum distance between a pair of neighboring agents in $p$. For a positive number $d > 0$, 
we define a subset of $P_G$ as follows: 
\begin{equation}\label{eq:defQd}
Q_G(d):= \left\{ p\in P_G \mid d_-(p) \ge d  \right \}.
\end{equation}
We now establish the following fact:

\begin{lem}\label{lem:evaluatedifferenceoftwovectorfields}
Let $d> 0$ be a fixed number. Then,  for any $\epsilon > 0$, there is a $\delta > 0$ such that if $p$ and $p'$ are in $Q_G(d)$ with $\|p - p'\| \le \delta$, 
then   
$\|f(p) - f(p')\| \le \epsilon$. 
\end{lem}

\begin{proof}
Let 
$p = (x_1,\ldots, x_N)$ and $p' = (x'_1,\ldots, x'_N)$. Denote by
$d_{ij}:= \|x_j - x_i\|$ and  $d'_{ij}:= \|x'_j - x'_i\|$. 
Note that 
$$
\|f(p) - f(p')\|^2 = \sum_{v_i\in V}\|f_i(p) - f_i(p')\|^2, 
$$
with each term $\|f_i(p) - f_i(p')\|$ bounded above by
$$ \sum_{v_j\in V_i} \| g_{ij}(d_{ij})(x_j - x_i) - g_{ij}(d'_{ij})(x'_j - x'_i) \|.$$
We also note that if $\|p - p'\| < \delta$, then 
$$ \|(x_j - x_i) - (x'_j - x'_i) \| < 2\delta.$$ 
It thus suffices to show that for any $\epsilon' > 0$, there is a $\delta' > 0$ such that if two vectors $u$, $u'\in \R^n$ satisfy 
\begin{equation}\label{eq:forcontinuity0}
\min\{ \|u\|,  \|u'\|\} \ge d \hspace{10pt} \mbox{ and } \hspace{10pt} \|u - u'\| < \delta', 
\end{equation}
then, for all $(v_i,v_j) \in E$, we have
\begin{equation}\label{eq:forcontinuity1}
\|g_{ij}(\|u\|) u - g_{ij}(\|u'\|) u' \| < \epsilon'. 
\end{equation} 

Recall that $\ol g_{ij}(d)$ (defined in~\eqref{eq:defbargij}) is given by $\ol g_{ij}(d) = dg_{ij}(d)$. From the condition of {\it fading attraction}, there exists a number $d_*$ such that if $\|u\| \ge d_*$, then
$ \ol g_{ij}(\|u\|) < \epsilon' / 2$ for all $(v_i,v_j)\in E$.   
Hence,  if $\min\{\|u\|,  \|u'\| \}\ge d_*$, then,  
$$
\|g_{ij}(\|u\|) u - g_{ij}(\|u'\|) u' \| \le \ol g_{ij}(\|u\|) + \ol g_{ij}(\|u'\|) < \epsilon'.
$$
Define a subset $K$ of $\R^n$ as follows:
$$
K:= \{u\in \R^n \mid d \le \|u\| \le d_* + 1\}.
$$
Since $K$ is compact and the map 
$$
\widetilde g_{ij}: u\mapsto g_{ij}(\|u\|) u
$$ 
is continuous, 
there exists a $\delta'\in (0,1)$ such that if $u$ and $u'$ are in $K$, with $\|u - u'\| < \delta'$, 
then 
$
\|\widetilde g_{ij}(u) - \widetilde g_{ij}(u') \| < \epsilon'
$  for all $(v_i,v_j) \in E$. 
Now, let $u, u'\in \R^n$ satisfy~\eqref{eq:forcontinuity0}, and let $\delta' < 1$. Then, either $\min\{\|u\|, \|u'\|\} \ge d_*$ or $\{u, u' \}\subset K$. Since~\eqref{eq:forcontinuity1} holds in either of the two cases, we complete the proof. 
\end{proof}

With Lemma~\ref{lem:evaluatedifferenceoftwovectorfields}, we prove below Proposition~\ref{pro:defnud}:

\begin{proof}[Proof of Proposition~\ref{pro:defnud}]
We first show that $\mu$ is continuous, and then show that $\mu(d) > 0$ for all $d\ge D_+$.

\noindent
{\it 1). Proof that $\mu$ is continuous}.  
We fix a distance $d>0$, and show that $\mu$ is continuous at $d$. Specifically, we show that for any $\epsilon > 0$, there is a $\delta > 0$ such that if $|d' - d|< \delta$, then 
$|\mu(d) - \mu(d')| < \epsilon$.

Let $p\in X_G(d)$ and  $B$ be a closed neighborhood of $p$ in $P_G$. Then, there is an open neighborhood $I$ of $d$ in $\R_+$ such that $X_G(d')$ intersects $B$ for all $d'\in I$. From Proposition~\ref{pro:nmvecfld}, there exists a $d_*$ such that if $p'' \in X_G(d)$, with $d_-(p'')< d_*$, then $\|f(p'')\| \ge \|f(p')\| $ for all $p'\in B$. 
This, in particular, implies that for all $d' \in I$, we have
\begin{equation}\label{eq:alihuiwoyoujian}
\inf\{ \|f(p)\| \mid p\in X_{G}(d') \cap Q_G(d_*) \} = \mu(d').
\end{equation}
Let $\delta>0$ be sufficiently small such that if $|d' - d| < \delta$, then $d'\in I$.

Choose a number $d'$ with $|d' - d| < \delta$;  without loss of generality, we assume that $\mu(d) \le \mu(d')$. 
From~\eqref{eq:alihuiwoyoujian}, there is a sequence $\{p(i)\}_{i\in\N}$, with each $p(i)\in X_{G}(d)\cap Q_G(d_*)$,  such that
$
\lim_{i\to \infty} \|f(p(i))\| = \mu(d)
$. 
Note that if $\delta$ is chosen sufficiently small, then for each $p(i)$ in the sequence, there exists a $p'(i)$ in the intersection of $X_{G}(d')$ and $Q_G(d_*/ 2)$ such that 
$
\|p'(i) -p(i)\| = |d' - d|
$.  
To see this, let $p(i) = (x_1(i),\ldots, x_N(i))$; without loss of generality,  we assume that 
$\|x_2(i) - x_1(i)\| = d$. We then set $p'(i)$ as follows: let 
$$
x'_1(i) := x_1(i) + (d'/ d - 1) (x_1(i) - x_2(i)),
$$
and let $x'_j(i) := x_j(i) $ for $v_j \neq v_1$.  
Then, by construction, $\|x'_2(i) - x'_1(i)\| = d'$, and 
$$ \|p'(i) - p(i)\| = \|x'_1(i) - x_1(i)\| = |d' - d| < \delta.$$
Moreover, if we let $\delta < d_* /2$, then from the fact that $p(i)\in Q_G(d_*)$, we have $p'(i) \in Q_G(d_*/2)$. 


From Lemma~\ref{lem:evaluatedifferenceoftwovectorfields}, we can choose 
 $\delta$ sufficiently small such that if $p$ and $p'$ are in $Q_G(d_*/2)$, with $\|p' - p\| \le \delta$, then $\|f(p) - f(p')\| \le \epsilon$. Since $\|p(i) - p'(i)\| = |d' - d| < \delta$, we have 
\begin{equation}\label{eq:boundedfpi}
\| f(p'(i)) - f(p(i)) \| \le \epsilon, \hspace{10pt} \forall\, i\in \N.
\end{equation}
This, in particular, implies that the sequence $\{\| f(p'(i)) \|\}_{i\in \N}$ is bounded, and hence there is a converging subsequence  $\{\|f(p'(j_i))\|\}_{i\in\N}$. We thus let 
$
\mu':= \lim_{i\to \infty} \| f(p'(j_i)) \| 
$. By definition, we have $\mu' \ge \mu(d')$. On the other hand, we also have $\mu(d') \ge \mu(d)$, and hence 
$
0\le \mu(d') - \mu(d) \le \mu' - \mu(d)
$. 
Furthermore, from~\eqref{eq:boundedfpi}, we have 
$
\mu' - \mu(d) \le \epsilon 
$, and hence $\mu(d') - \mu(d) \le \epsilon$.  
This establishes the continuity of~$\mu$. 

\vspace{3pt}
\noindent
{\it 2). Proof that $\mu(d) > 0$ for all $d\ge D_+$}.
The proof is carried out by induction on the number of vertices of $G$. For the base case $N = 2$, we have $D_+ = \alpha_+$. The set $X_G(d)$ is nothing but
$$
X_G(d) = \left \{  (x_1, x_2)\in \R^4  \mid \|x_2 - x_1\| = d \right \}.
$$
So then, for all $d > \alpha_+$, we have $\mu(d) = \sqrt{2} \, | \ol g_{12}(d) | > 0$. 

For the inductive step, we assume that the statement holds for all $N \le (k-1)$, and prove for $N = k$.  We first have some notations. 
Let $G' = (V', E')$ be a connected, proper subgraph of $G$. Let $S'$ be the sub-system induced by $G'$, and $f'(p')$ be the associated vector field of $S'$ at $p'\in P_{G'}$. 
Similarly, let 
$$
X_{G', ij}(d) := \left\{ p'\in P_{G'} \mid  \|x_{j} - x_{i}\| = d \right\};
$$
and let $X_{G'}(d) := \cup_{(v_i,v_j)\in E'} X_{G', ij}(d)$.  We then define
$$
\mu_{G'}(d) := \inf\left\{ \|f'(p')\| \mid p' \in X_{G'}(d) \right\}.
$$
Note that if $d > (|V| -1) \, \alpha_+$, then   
$
d > (|V'| - 1) \, \alpha_+
$.  Thus, we can appeal to the induction hypothesis and obtain $\mu_{G'}(d) > 0$. We further define   
$$
\nu(d) := \min_{G'}\{ \mu_{G'}(d) \}, 
$$ 
where the minimum is taken over all  connected proper subgraphs of $G$. Then, $\nu(d) > 0$ for all $d > D_+ $.

We now fix a number $d > D_+ $, and prove that $\mu(d) > 0$. First, from Proposition~\ref{pro:nmvecfld}, there exists a $d_0 > 0$ such that 
\begin{equation}\label{eq:proveforsmallconfiguration}
\| f(p) \| > 1, \hspace{10pt} \mbox{ if } \hspace{5pt} p \in X_{G}(d) \mbox{ and }  d_-(p) < d_0.
\end{equation}
We also claim that there exists a  $d_1 > 0$ such that 
\begin{equation}\label{eq:proveforlargeconfiguration}
\|f(p)\| > \nu(d) /2, \hspace{10pt} \mbox{ if } \hspace{5pt} p \in X_{G}(d) \mbox{ and }  d_+(p) > d_1.
\end{equation}
Note that if this holds, then the proof is complete. Indeed, let $K$ be a subset of $X(d)$ defined as follows: 
$$
K := \left\{ p\in X_G(d) \mid d_0 \le \|x_j - x_i\| \le d_1, \, \forall (v_i,v_j) \in E\right\}.
$$
It is known that system~\eqref{MODEL} is an equivariant system with respect to the special Euclidean group. In particular, $\|f(p)\| = \|f(p')\|$ if $p$ and $p'$ are related by translation and/or rotation. On the other hand,  $K$ is compact modulo translation and rotation, and moreover, $f(p)$ does not vanish over $K$. We thus have that 
$
\inf_{p\in K} \| f(p) \|> 0
$. 
Combining this fact with~\eqref{eq:proveforsmallconfiguration} and~\eqref{eq:proveforlargeconfiguration}, we obtain $\mu(d) > 0$. 

It thus remains to show that there exists a $d_1 > 0$  such that~\eqref{eq:proveforlargeconfiguration} holds. 
First, from the condition of {\it fading attraction}, there exists an $l > 0$ such that for all $(v_i,v_j) \in E$, we have  
\begin{equation}\label{eq:verylargedistance}
0 < \ol g_{ij}(d') < \nu(d) / (2k^2), \hspace{10pt} \, \forall\, d' \ge l. 
\end{equation}
Without loss of generality, we assume that $l$ is large enough so that $l > d$. 
We now define $d_1$ as follows: let $d_1$ be such that if a configuration $p\in X_{G}(d)$ satisfies $d_+(p) > d_1$; then,  
there is a nontrivial partition $\sigma$ of $(G,p)$ in $\Sigma(l\; p)$. Note that  from Remark~\ref{rmk:2}, such number $d_1$ exists. We show below that~\eqref{eq:proveforlargeconfiguration} holds for this choice of $d_1$. 

Let $(v_i, v_j) \in E$  be chosen such that $\|x_j - x_i\| = d$. Since $l> d$, there is a sub-framework $(G',p')$ associated with the partition $\sigma$ such that 
 $v_{i}$ and $v_j$ are vertices of  $G'$ (and $x_i$ and $x_j$ are agents in $p'$). 
Let $S'$ be the sub-system induced by $G'$, and $f'(p')$ be the vector field associated with $S'$ at~$p'$. 
Since $p'\in X_{G'}(d)$ and $G'$ is a proper subgraph of $G$, by definition of $\nu(d)$, we have $\|f'(p')\| \ge \nu(d)$. 
Let $V'$ be the vertex set of $G'$; without loss of generality, we assume that $V'= \{v_1,\ldots, v_{k'}\}$, with $k' < k$. 
Define a vector 
$
h := \(h_1,\ldots, h_{k'}\) \in \R^{nk'} 
$, 
with each $h_i\in \R^n$ given by
$$
h_i := \sum_{v_j\in V_i - V'} g_{ij}(d_{ij}) (x_j - x_i). 
$$
By the fact that  $\sigma$ is in $\Sigma(l\; p)$, we have $d_{ij} \ge l$ for $v_i\in V'$ and $v_j\notin V'$. Appealing to~\eqref{eq:verylargedistance}, we obtain 
$$
\|h_i\| \le \sum_{v_j\in V_i - V'} \ol g_{ij}(d_{ij}) < \nu(d) / (2k), 
$$
which implies that $\|h\| < \nu(d) / 2$. 
Recall that $f_{V'}(p)$ is the restriction of $f(p)$ to $p'$. By construction, we have
$
f_{V'}(p) =  f'(p') + h 
$, and hence 
$$
\|f(p)\| \ge \|f_{V'}(p)\| \ge \|f'(p')\| - \|h\| > \nu(d) /2.
$$
We have thus established~\eqref{eq:proveforlargeconfiguration}. This completes the proof. 
\end{proof}

\subsection{Proof of Proposition~\ref{pro:relatepartitiontodynamicalsystem}}
Recall that the subset $X_{G,ij}(d)$ (defined in~\eqref{eq:defXGijd}) is given by  
$$
X_{G,ij}(d) = \left\{ p\in P_G \mid \|x_j - x_i\| = d \right\}. 
$$
Now, let $d'$ and $d''$ be two positive numbers, and let $d(X_{G,ij}(d'),  X_{G,ij}(d''))$ be the distance between $X_{G,ij}(d')$ and $X_{G,ij}(d'')$, which is defined as follows: 
$$
\inf\left\{ \|p' - p''\| \mid p'\in X_{G,ij}(d'), \, p''\in X_{G,ij}(d'') \right\}. 
$$
We have the following fact:

\begin{lem}\label{lem:boundedbelowdistance}
Let $d'$ and $d''$ be two positive numbers. Then, 
$$
d(X_{G,ij}(d'), X_{G,ij}(d'')) = |d' - d''|/ \sqrt{2}. 
$$\,
\end{lem}

\begin{proof}
First, note that there are configurations $p' \in X_{G,ij}(d')$ and $p''\in X_{G,ij}(d'')$ such that
\begin{equation*}\label{eq:equalitysatisfied}
\|p' - p''\| =  |d' - d''|/ \sqrt{2}.
\end{equation*} 
Indeed, let $p' = (x'_1,\ldots, x'_N)$ and $p'' = (x''_1,\ldots, x''_N)$; we then set
$$
\left\{
\begin{array}{l}
x'_i  =  -x'_j   =  (d', 0, \ldots, 0) / 2, \vspace{1pt}\\
x''_i  =  - x''_j =  (d'',0,\ldots,0) / 2.  \vspace{1pt}\\
\end{array}
\right.
$$ 
For the other agents, we set $x'_k  =  x''_k$  for all  $v_k \in V-\{v_i, v_j\}$, 
subject to the constraint that $x'_{a} \neq x'_{b}$ and $x''_{a}\neq x''_{b}$, for all $(v_{a},v_{b})$ in $E$.

We now show that if $p'\in X_{G,ij}(d')$ and $p'' \in X_{G,ij}(d'')$, then 
$$ \|p' - p''\| \ge |d' - d''|/ \sqrt{2}.$$ 
It suffices to show that 
\begin{equation}\label{eq:lem81}
 \|x'_i - x''_i\|^2 + \|x'_j - x''_j\|^2 \ge \frac{1}{2} (d' - d'')^2.
\end{equation}
Let $x':= (x'_i + x'_j) / 2$,  $x'':= (x''_i + x''_j) / 2$, 
and let 
$$
\left\{
\begin{array}{ll}
y'_i := x'_i - x' & y'_j := x'_j - x', \vspace{3pt}\\
y''_i := x''_i - x'' & y''_j := x''_j - x''.  
\end{array}
\right.
$$
First, note that $ y'_i + y'_j = y''_i + y''_j = 0$, and hence
\begin{equation}\label{eq:lem82}
\|x'_i - x''_i\|^2 + \|x'_j - x''_j\|^2 = 2 \|y'_i - y''_i\|^2 + \|x' - x''\|^2.
\end{equation}
We also note that $\|y'_i\| = d'/2 $ and $\|y'_i\| = d''/2$, and hence by the triangle inequality, 
\begin{equation}\label{eq:lem83}
 \|y'_i - y''_i\| \ge | \|y'_i\| - \|y''_i \||  = |d' - d''|/ 2.  
\end{equation}
Combining~\eqref{eq:lem82} and~\eqref{eq:lem83}, we then establish~\eqref{eq:lem81}. 
\end{proof}

To prove Proposition~\ref{pro:relatepartitiontodynamicalsystem}, we further need the following fact:

\begin{lem}\label{lem:upperboundonvelocity}
Let $p(t)$ be a trajectory generated by system~\eqref{MODEL}. Then, the following hold:
\begin{enumerate}
\item $\sup_{t \ge 0} \|f(p(t)\|  < \infty$.
\item For any $\epsilon > 0 $, there exists an instant $T_{\epsilon}$ such that 
\begin{equation}\label{eq:defTepsilon}
\Psi\(p(T_{\epsilon})\) - \Psi\(p(\infty)\) \le \epsilon.
\end{equation}
\end{enumerate}\,
\end{lem}

\begin{proof} We first prove part 1. 
Let $d_{ij}(t):= \|x_j(t) -x_i(t)\|$. It suffices to show that for all $(v_i,v_j) \in E$,  
\begin{equation}\label{eq:Novsecond8:45am}
\sup \{d_{ij}(t)g_{ij}(d_{ij}(t)) \mid t \ge 0   \} < \infty.
\end{equation}
From Lemma~\ref{ELB}, there is a number $d_*> 0$ such that $d_-(p(t)) \ge d_*$ for all $t \ge 0$. 
Then, from the condition of {\it fading attraction}, we have
$$
\sup\{dg_{ij}(d) \mid d \ge d_*\} < \infty,
$$
which implies that~\eqref{eq:Novsecond8:45am} holds. 
We now prove part 2. From Lemma~\ref{lem:phiboundedbelow}, the potential function $\Psi$ is bounded below. On the other hand, $\Psi(p(t))$ is non-increasing. Hence,   
the limit
$
\Psi(q(\infty)) := \lim_{t \to \infty}\Psi(q(t)) 
$   
exists, which then implies the existence of $T_{\epsilon}$ such that~\eqref{eq:defTepsilon} holds.   
\end{proof}

With Lemmas~\ref{lem:boundedbelowdistance} and~\ref{lem:upperboundonvelocity}, we prove Proposition~\ref{pro:relatepartitiontodynamicalsystem}.

\begin{proof}[Proof of Proposition~\ref{pro:relatepartitiontodynamicalsystem}]
Fix an~$i\in \N$; we prove Proposition~\ref{pro:relatepartitiontodynamicalsystem} by first exhibiting a~$j'_i\in \N$ such that 
\begin{equation}\label{eq:firstproof}
\cal{L}_-(\sigma_t) < L_0, \hspace{10pt} \forall\,  t\ge t_{j'_i}, 
\end{equation}
and then, exhibiting a~$j''_i\in \N$ such that 
\begin{equation}\label{eq:secondproof}
\cal{L}_+(\sigma_t) > \max\{l_i, L_0\},  \hspace{10pt} \forall\, t \ge t_{j''_i}.  
\end{equation}
Note that if such indices~$j'_i$ and~$j''_i$ exist, then the proof is complete; indeed, let $j_i := \max\{j'_i, j''_i\}$, then $\sigma_t$ is in $\Sigma(l_i\; q(t))$ for all $t \ge t_{j_i}$. 
We now establish~\eqref{eq:firstproof} and~\eqref{eq:secondproof}, respectively.

\vspace{3pt}
\noindent
{\it 1). Proof of existence of~$j'_i$}. We first make some definitions. 
By assumption, we have $L_0 > D_+$, 
and hence from Proposition~\ref{pro:defnud}, $\mu(L_0) > 0$. Since $\mu$ is continuous,  there exists a $\delta_0 > 0$ such that if we let $I_0: = [L_0 - \delta_0, L_0+ \delta_0]$, then    
$\mu(d) \ge \mu(L_0)/2$ for all  $d \in I_0$.  
Let $\xi:= \sup_{t \ge 0} \| f(q(t)) \|$, which is a positive real number by Lemma~\ref{lem:upperboundonvelocity}, and let 
$
\tau_0 := {\delta_0} / (\sqrt{2} \, \xi )
$.

Let 
$
X_{G,ij}(I_0) := \cup_{d\in I_0} \, X_{G,ij}(d)
$. 
We show below that if, at certain instant $t_0\ge 0$, we have  $q(t_0)\in X_{G,ij}(L_0)$ for some $(v_i,v_j) \in E$; then,   
\begin{equation}\label{eq:evaluatept'}
q(t) \in X_{G,ij}(I_0), \hspace{10pt} \forall\, t \in  [t_0, t_0 +  \tau_0 ]. 
\end{equation}
First, note that if the trajectory $q(t)$ leaves $X_{G,ij}(I_0)$ at $t' > t_0$, then it {\it has to} intersect either $ X_{G,ij}(L_0 + \delta_0)$ or  $ X_{G,ij}(L_0 - \delta_0)$. 
On the other hand, from Lemma~\ref{lem:boundedbelowdistance}, 
if $p', p''\in P_G$ are such that $p'\in X_{G,ij}(L_0)$ and 
$
p''\in X_{G,ij}(L_0 \pm\delta_0)
$, 
then, 
$
\|p' - p''\| \ge \delta_0/ \sqrt{2}
$.  Furthermore, we have
$\|\dot q(t)\| \le \xi$ for all $t \ge 0$. Hence, starting from $q(t_0)\in X_{G,ij}(L_0)$, the trajectory $q(t)$ has to remain within $X_{G,ij}(I_0)$ for at least $\tau_0$ units of time. We have thus established~\eqref{eq:evaluatept'}. 

On the other hand, we have
$$
\Psi(q(t_0)) - \Psi(q(t_0 + \tau_0))   =  \int^{t_0+ \tau_0}_{t_0}\| f(q(t)) \|^2 \, dt. 
$$
Combining~\eqref{eq:evaluatept'} with the fact that $f(p) \ge \mu(L_0) /2 $ for all $p\in X_{G,ij}(I_0)$, we obtain 
\begin{equation*}\label{eq:epsilon0decrease}
\Psi(q(t_0)) - \Psi(q(t_0 + \tau_0))   \ge \epsilon_0, \hspace{5pt}\mbox{for} \hspace{5pt} \epsilon_0 :=  \frac{\mu(L_0)^2\tau_0}{4}.
\end{equation*}
From Lemma~\ref{lem:upperboundonvelocity}, there is an instant $T_{\epsilon_0}$ such that 
\begin{equation*}\label{eq:deltaisepsilon0}
\Psi(q(T_{\epsilon_0})) - \Psi(q(\infty)) = \epsilon_0. 
\end{equation*}
Since the sequence $\{t_i\}_{i\in \N}$  monotonically increases, and approaches to infinity, there is a~$j'_i\in \N$ such that 
$
t_{j'_i} > T_{\epsilon_0}  
$. 
We now show that~\eqref{eq:firstproof} holds for the choice of~$j'_i$. The proof is carried out by contradiction. Suppose that, to the contrary, there is an instant $t_0 \ge t_{j'_i}$ such that
$
L_{-}(\sigma_{t_0}) = L_0
$. Then, $q(t_0) \in X_{G}(L_0)$, and hence by the arguments above, we have $\Psi(q(t_0 + \tau_0))\le \Psi(q(t_0)) - \epsilon_0$. Moreover, since $t_0 \ge t_{j'_i} > T_{\epsilon_0}$, and by the fact that $\Phi(q(t))$ strictly monotonically decreases in~$t$, we have
$ \Psi(q(t_0)) < \Psi(q(T_{\epsilon_0}))$. Combining these facts, we obtain 
$$
\Psi(q(t + \tau_0)) < \Psi(q(T_{\epsilon_0})) - \epsilon = \Psi(q(\infty)),
$$
which is a contradiction. We have thus shown that~\eqref{eq:firstproof} holds for the choice of~$j'_i$.

\vspace{3pt}
\noindent
{\it 2). Proof of existence of~$j''_i$}. The proof here is similar to the proof of existence of~$j'_i$.  Let
$L_1:=\max\{l_i, L_0\}$;  
from Proposition~\ref{pro:defnud}, there is a closed interval 
$I_1:= [L_1 - \delta_1, L_1+ \delta_1]$, for some $\delta_1 > 0$,  
such that
$\mu(d) \ge \mu(L_1)/2$ for all $d \in I_1$.  Let $\tau_1 := {\delta_1} / (\sqrt{2} \, \xi )$. Suppose that at certain instant $t_1$, $q(t_1)\in X_{G,ij}(L_1)$ for some $(v_i,v_j)\in E$,  then from Lemma~\ref{lem:boundedbelowdistance}, 
$
q(t) \in  X_{G,ij}(I_1) 
$  for all  $t \in  [t_1, t_1 +  \tau_1 ]$. 
It then follows that 
$$
\Psi(q(t_1)) - \Psi(q(t_1 + \tau_1))   \ge \epsilon_1, \hspace{5pt} \mbox{for} \hspace{5pt} \epsilon_1 := \frac{\mu(L_1)^2\tau_1}{4}.
$$
Appealing again to Lemma~\ref{lem:upperboundonvelocity}, we obtain an instant $T_{\epsilon_1}$ such that 
$
\Psi(q(T_{\epsilon_1})) - \Psi(q(\infty)) = \epsilon_1 
$. 
Since both sequences $\{t_i\}_{i\in \N}$ and $\{l_i\}_{i\in \N}$ monotonically increase, and approach to infinity, there is a~$j''_i\in \N$ such that 
$
t_{j''_i} > T_{\epsilon_1}$ and  $l_{j''_i} > L_1$.  Then,~\eqref{eq:secondproof} holds for the choice of~$j''_i$ because otherwise, there will be an instant~$t_1$, with $t_1 \ge t_{j''_i}$, such that $L_+(\sigma_{t_1}) = L_1$, and hence
$\Psi(q(t_1 + \tau_1)) < \Psi(q(\infty))$, which is a contradiction. 
 This completes the proof.   
\end{proof}

\end{document}